\documentclass[11pt]{article}
\usepackage[letterpaper, margin=1in]{geometry}
\usepackage{mathtools, mleftright, amssymb, amsthm, physics, tikz, ifthen, multirow, adjustbox, array, float, soul, bbm, makecell, nicefrac, booktabs, subcaption, fontawesome, graphicx, bbold, authblk}
\usepackage{todonotes}
\usepackage[usestackEOL]{stackengine}

\usepackage{hyperref}
\hypersetup{
    colorlinks,
    linkcolor={red!80!black},
    citecolor={blue},
    urlcolor={blue!80!black}
}
\usepackage{cleveref}

\newcommand*\cno[1][0.8ex]{\tikz\draw (0,0) circle (#1);}
\newcommand*\cboth[1][0.8ex]{%
  \begin{tikzpicture}
  \draw[fill] (0,0) -- (180:#1) arc (180:360:#1) -- cycle;
  \draw (0,0) circle (#1);
  \end{tikzpicture}}
\newcommand*\cyes[1][0.8ex]{\tikz\fill (0,0) circle (#1);}

\newcommand{\drop}[1]{}
\usepackage{algorithm}
\usepackage{algorithmicx}
\usepackage{algpseudocode}
\algtext*{EndWhile}
\algtext*{EndIf}
\algtext*{EndFor}
\algrenewcommand\algorithmicrequire{\textbf{Input:}}
\algrenewcommand\algorithmicensure{\textbf{Output:}}

\newcommand{\Deg}{\textnormal{\texttt{DEG}}}
\newcommand{\Neigh}{\textnormal{\texttt{NEIGH}}}
\newcommand{\Jump}{\textnormal{\texttt{JUMP}}}
\newcommand{\NeighSorted}{\textnormal{\texttt{NEIGH-SORTED}}}
\newcommand{\Adj}{\textnormal{\texttt{ADJ}}}

\newcommand{\pw}{\textnormal{\texttt{PowerMethod}}}

\newcommand{\RW}{\textnormal{\texttt{RandomWalk}}}
\newcommand{\RBS}{\textnormal{\texttt{RandPush}}}
\newcommand{\BackPush}{\texttt{BacwardsPush}}
\newcommand{\BiPPRAvg}{\texttt{BiPPRAvg}}
\newcommand{\BackPushNew}{\textnormal{\texttt{BackwardsPushAvg}}}

\newcommand{\SingleNode}{\textnormal{\texttt{SingleNode}}}

\theoremstyle{plain}
\newtheorem{theorem}{Theorem}[subsection]
\newtheorem{lemma}[theorem]{Lemma}

\theoremstyle{definition}
\newtheorem{definition}[theorem]{Definition}

\theoremstyle{remark}

\newcommand{\p}[1]{\mleft( #1 \mright)}
\newcommand{\np}[1]{( #1 )}

\newcommand{\curly}[1]{\mleft\{ #1 \mright\}}
\newcommand{\ncurly}[1]{\{ #1 \}}
\renewcommand{\square}[1]{\mleft[ #1 \mright]}

\newcommand{\floor}[1]{\mleft\lfloor #1 \mright\rfloor}
\newcommand{\plainfrac}[2]{#1/#2}

\newcommand{\mc}{\mathcal}
\renewcommand{\P}[1]{\mathbb{P} \square{ #1 } }
\newcommand{\E}[1]{\mathbb{E} \square{ #1 } }
\newcommand{\Var}[1]{\text{Var} \square{ #1 } }
\newcommand{\Cov}[1]{\text{Cov} \p{ #1 } }
\newcommand{\Ind}[1]{\mathbbm{1} \curly{ #1 } }

\newcommand{\pih}{\hat{\pi}}
\newcommand{\residue}{r}
\newcommand{\reserve}{p}
\newcommand{\rh}{\hat{r}}
\newcommand{\ph}{\hat{p}}
\newcommand{\rp}{r'}
\newcommand{\pp}{p'}

\newcommand{\highlight}[1]{{\color{red} #1}}
\newcommand{\hcolor}{\highlight{red} }

\title{Personalized PageRank Estimation in Undirected Graphs}

\author{
  Christian Bertram\thanks{\texttt{chbe@di.ku.dk}, \href{https://orcid.org/0009-0009-7940-1002}{ORCID 0009-0009-7940-1002}}
  \ and
  Mads Vestergaard Jensen\thanks{\texttt{mvje@di.ku.dk}, \href{https://orcid.org/0009-0007-0929-5924}{ORCID 0009-0007-0929-5924}}
}

\affil{BARC, University of Copenhagen\thanks{The authors are part of BARC, Basic Algorithms Research Copenhagen, supported by the VILLUM Foundation grant 54451.}}

\date{}

\begin{document}
\maketitle

\pagenumbering{roman}
\setcounter{page}{1}
\begin{abstract}
Given an undirected graph $G=(V, E)$, the Personalized PageRank (PPR) of $t\in V$ with respect to $s\in V$, denoted $\pi(s,t)$, is the probability that an $\alpha$-discounted random walk starting at $s$ terminates at $t$.
We study the time complexity of estimating $\pi(s,t)$ with constant relative error and constant failure probability, whenever $\pi(s,t)$ is above a given threshold parameter $\delta\in(0,1)$.
We consider common graph-access models and furthermore study the single source, single target, and single node (PageRank centrality) variants of the problem.

We provide a complete characterization of PPR estimation in undirected graphs by giving tight bounds (up to logarithmic factors) for all problems and model variants in both the worst-case and average-case setting.
This includes both new upper and lower bounds.
    
Tight bounds were recently obtained by Bertram, Jensen, Thorup, Wang, and Yan \cite{DirectedPPR} for directed graphs.
However, their lower bound constructions rely on asymmetry and therefore do not carry over to undirected graphs.
At the same time, undirected graphs exhibit additional structure that can be exploited algorithmically.
Our results resolve the undirected case by developing new techniques that capture both aspects, yielding tight bounds.
\end{abstract}

\newpage
\tableofcontents
\newpage
\pagenumbering{arabic}
\setcounter{page}{1}
\section{Introduction}\label{sec:introduction}

Estimation of \emph{Personalized PageRank (PPR)} is a standard tool for ranking vertices in a graph.
Given an undirected graph $G=(V,E)$ with $n$ vertices, the PPR $\pi(s,t)$ of a target $t \in V$ with respect to a source $s \in V$ is defined as the probability that an $\alpha$-discounted random walk starting at $s$ terminates at $t$.
An $\alpha$-discounted random walk is a random walk which in each step terminates with probability $\alpha$, and otherwise moves to a uniformly random neighbor.
Following prior work, we assume $\alpha$ to be a constant~\cite{thorup2025pagerank, STOC-BIPPR}.
If we choose the source $s$ uniformly at random, the probability of terminating at $t$ is $\pi(t)=\sum_s\pi(s,t)/n$, which is referred to as the \emph{PageRank centrality} of $t$.
This global ranking notion is the quantity studied in the seminal work of Page, Brin, Motwani, and Winograd~\cite{page1999pagerank} in the context of ranking web pages.
These PageRank problems have been extensively studied for undirected graphs, with some of the most important results in~\cite{DirectedPPR, UndirectedBiPPR, Lofgren_backpush, BackMC, RBS}.
PageRank also has numerous applications in information retrieval~\cite{topic-sensitive, page1999pagerank}, recommendation systems~\cite{robust-reputations, eigenvalue-trust}, machine learning~\cite{clustering, supervised-learning}, and graph partitioning~\cite{local-partitioning, heat-patitioning}.

In this work, we study the computational complexity of estimating $\pi(s,t)$.
Concretely, for an approximation threshold $\delta\in (0,1]$, we require an estimate of $\pi(s,t)$ that is within a constant relative error whenever $\pi(s,t)>\delta$, with constant failure probability.
If instead $\pi(s,t)\leq \delta$, we only require an additive error of $O(\delta)$. 
We formalize this estimation requirement later in the introduction.
We further consider the following query problems from the litterature:
\begin{itemize}
  \setlength\itemsep{2pt}
    \item \emph{Single pair:} Given a pair $(s,t)\in V^2$, estimate $\pi(s,t)$,
    \item \emph{Single source:} Given a source $s\in V$, estimate $\pi(s,t)$ for each target $t\in V$,
    \item \emph{Single target:} Given a target $t\in V$, estimate $\pi(s,t)$ for each source $s\in V$, and
    \item \emph{Single node:} Given a vertex $t \in V$, estimate $\pi(t)$.
\end{itemize}
We measure the computational complexity in the standard \emph{adjecency-list model}~\cite{DirectedPPR, STOC-BIPPR}, where the algorithm can access the input graph only through an oracle supporting the query $\Deg(v)$ to obtain the degree of $v$ and $\Neigh(v,i)$ to obtain the $i$-th neighbor of $v$, each in constant time.
These queries can only be applied to vertices given in the input or returned by previous queries.
Following the literature, we also consider extensions of this model in which the algorithm additionally has access to any subset of the following queries: 
$\Jump()$ which returns a uniformly random vertex~\cite{thorup2025pagerank, STOC-BIPPR},
$\NeighSorted(v,i)$ which returns the $i$-th neighbor of $v$ when neighbors are ordered by degree in increasing order~\cite{DirectedPPR, RBS}, and $\Adj(u,v)$ which tests whether $uv \in E$~\cite{DirectedPPR}.
In the literature, the adjacency-list model with $\Jump$ is referred to as the \emph{graph access model}~\cite{STOC-BIPPR}.

We also consider \emph{average-case} complexity, in the following sense.
For the single pair problem, this measure the expected running time when \((s,t)\) is drawn uniformly at random from \(V^2\).
For the single source, single target, and single node problems, this measures the expected running time under a uniformly random choice of the designated source or target vertex, respectively.
Note that the graph is still ``worst-case'' and that the algorithm must still be correct (with constant probability) for every source/target.

Recently, all of the above problems and model variants were studied for directed graphs, where tight bounds were established in~\cite{DirectedPPR}.
A main goal of this work is to understand how these bounds change in the undirected setting.
In~\cite{DirectedPPR}, they provide lower bounds that exploit asymmetry, and these do not directly carry over to undirected graphs.
In contrast, undirected graphs satisfy the \emph{reversibility property}~\cite{UndirectedBiPPR}, which relates $\pi(s,t)$ and $\pi(t,s)$, by the identity
\begin{align}
    d(s)\pi(s,t)=d(t)\pi(t,s),\label{eq:reverse-intro}
\end{align}
where $d(s)$ and $d(t)$ are the degrees of $s$ and $t$, respectively.
This gives additional structure in the undirected case and can be exploited algorithmically.
For instance, it excludes a situation where $\pi(s,t)=\Omega(1)$ and $\pi(t,s)=0$, which is possible in directed graphs.

\paragraph{Our result: A complete characterization.}
We give tight bounds, up to polylogarithmic factors, for all problems (single pair, single source, single target, and single node) on undirected graphs, in the adjacency-list model with every combination of additional queries from $\{\Jump, \allowbreak~\NeighSorted, ~\Adj\}$.
An overview of all bounds can be found in~\Cref{tab:results}.
The bounds are stated as functions of the number of vertices $n = \abs{V}$, number of edges $m=\abs{E}$, approximation threshold $\delta$, and average degree $d=m/n$.
It should be noted, that some of the known upper bounds have been parameterized in $d(t)$ in prior work, but in this work, we do not include this parameter when constructing the lower bounds.
Overall, \Cref{tab:results}, gives a complete picture of the computational landscape of (Personalized) PageRank estimation in undirected graphs, and shows how the different query types affect the computational complexity.
This could be important when deciding which queries to support when designing an application programming interface (API) for large graphs.

\begin{table}[!h]
    \centering
    \renewcommand{\arraystretch}{1.4}
    \hspace{-0.3cm}
    \scalebox{1}{
    \begin{adjustbox}{center}
    \begin{tabular}{|c|c|p{0.2cm}|p{0.2cm}|p{0.2cm}|c|c|c|}
        \hline
        \multirow{2}{*}{\textbf{Problem}} & \multirow{2}{*}{\textbf{Case}} & \multicolumn{3}{c|}{\textbf{Model}} & \multirow{2}{*}{\textbf{Ours}} & \multicolumn{2}{c|}{\textbf{Previous best}}\\ \cline{3-5}\cline{7-8}
        & & \textbf{J} & \textbf{S} & \textbf{A} & & \textbf{Lower} & \textbf{Upper}\\ \hline\hline
        \multirow{4}{*}{\makecell{Single\\pair}} & Worst & \cboth & \cboth & \cboth & $\tilde\Theta\np{\min\ncurly{m,(\plainfrac{n}{\delta})^{1/2},\plainfrac{1}{\delta}}}$  & \makecell{$\Omega(\min\{n, 1/\delta\})$\\$\star$}  & \makecell{$\tilde O\np{\min\ncurly{m,(\plainfrac{n}{\delta})^{1/2},\plainfrac{1}{\delta}}}$\\\cite{page1999pagerank, UndirectedBiPPR, MC}} \\\cline{2-8}
        & \multirow{3}{*}{Avg.} & \cboth & \cno & \cboth &  \multirow{2}{*}{$\tilde\Theta\np{\min\ncurly{m, (d/\delta)^{1/2}, \highlight{(1/\delta)^{2/3}+d}, 1/\delta}} $ } & \multirow{2}{*}{---} & \multirow{2}{*}{\makecell{$\tilde O\np{\min\ncurly{m, \np{\plainfrac{d}{\delta}}^{\plainfrac12}, \plainfrac{1}{\delta}}}$\\\cite{page1999pagerank, UndirectedBiPPR, MC}}} \\ \cline{3-5}
        & & \cboth & \cboth & \cno &  & & \\\cline{3-8}
        & & \cboth & \cyes & \cyes & $\tilde\Theta\np{\min\ncurly{m, \np{\plainfrac{d}{\delta}}^{\plainfrac12}, \np{\plainfrac{1}{\delta}}^{\plainfrac23}}}$  & --- & \makecell{$\tilde O\np{\min\ncurly{m, \np{\plainfrac{d}{\delta}}^{\plainfrac12}, \np{\plainfrac{1}{\delta}}^{\plainfrac23}}}$\\\cite{page1999pagerank, UndirectedBiPPR, DirectedPPR}}\\ \hline\hline
        \multirow{2}{*}{\makecell{Single\\source}} & Worst & \multirow{2}{*}{\cboth} & \multirow{2}{*}{\cboth} & \multirow{2}{*}{\cboth} & \multirow{2}{*}{$\tilde\Theta\np{\min\ncurly{m, \plainfrac{1}{\delta}}}$ } & \multirow{2}{*}{\makecell{$\Omega\np{\min\ncurly{n, \plainfrac{1}{\delta}}}$\\$\star$}} & \multirow{2}{*}{\makecell{$\tilde O\np{\min\ncurly{m, \plainfrac{1}{\delta}}}$~\cite{page1999pagerank, MC}}} \\ \cline{2-2}
        & Avg. & & & & & & \\ \hline\hline
        \multirow{6}{*}{\makecell{Single\\target}} & \multirow{3}{*}{Worst} & \cno & \cno & \cboth & $\tilde\Theta\np{\min\ncurly{m, \plainfrac{n}{\delta}}}$  & \multirow{3}{*}{\makecell{$\Omega(n)$\\ $\star$}} & \multirow{3}{*}{\makecell{$\tilde O(\min\ncurly{m,n/\delta})$~\cite{page1999pagerank, Lofgren_backpush}}} \\ \cline{3-6}
        & & \cyes & \cboth & \cboth & \multirow{2}{*}{$\tilde\Theta\np{\min\ncurly{m, \highlight{\plainfrac{n}{\delta^{1/2}}}}}$ } & & \\ \cline{3-5}
        & & \cboth & \cyes & \cboth & & & \\ \cline{2-8}
        & \multirow{3}{*}{Avg.} & \cno & \cno & \cboth & $\tilde\Theta\np{\min\ncurly{m, \plainfrac{d}{\delta}, \highlight{(1/\delta)^2+d}}}$  & \multirow{4}{*}{\makecell{$\Omega\np{\min\ncurly{n, \plainfrac{1}{\delta}}}$\\$\star$}}  & \makecell{$\tilde O\np{\min\ncurly{m, \plainfrac{d}{\delta}}}$~\cite{page1999pagerank, Lofgren_backpush}} \\  \cline{3-6}\cline{8-8}
        & & \cyes & \cno & \cboth & \makecell{$\tilde\Theta\np{\min\{m, \np{\plainfrac{m}{\delta}}^{1/2}, d/\delta,$\hspace{1.3cm}\,\\\,\hspace{1.3cm}$\highlight{(n/\delta)^{2/3}+d}, \highlight{(1/\delta)^2+d}\}}$} & & \makecell{$\tilde O\np{\min\ncurly{m, \np{\plainfrac{m}{\delta}}^{\plainfrac12}, \plainfrac{d}{\delta}}}$\\\cite{page1999pagerank, DirectedPPR, Lofgren_backpush}} \\ \cline{3-6}\cline{8-8}
        & & \cboth & \cyes & \cboth & $\tilde\Theta\np{\min\ncurly{m, \plainfrac{1}{\delta}}}$  & & $\tilde O\np{\min\ncurly{m, \plainfrac{1}{\delta}}}$~\cite{page1999pagerank, RBS} \\ \hline\hline
        \multirow{9}{*}{\makecell{Single\\node}} & \multirow{5}{*}{Worst} & \cno & \cno & \cno & \multirow{4}{*}{$\Theta(m^{1/2})$ } & $\Omega(m^{1/2})$~\cite{BackMC} &  \multirow{5}{*}{$O(m^{1/2})$~\cite{BackMC}} \\ \cline{3-5}\cline{7-7}
        & & \cno & \cboth & \cboth & & \multirow{4}{*}{---} & \\ \cline{3-5}
        & & \cboth & \cno & \cboth & & & \\\cline{3-5}
        & & \cboth & \cboth & \cno & & & \\\cline{3-6}
        & & \cyes & \cyes & \cyes & $\Theta(\highlight{n^{1/2}})$  & & \\\cline{2-8} 
        & \multirow{4}{*}{Avg.} & \cno & \cboth & \cboth & \multirow{3}{*}{$\Theta\np{d}$ } & \multirow{4}{*}{---} & \multirow{4}{*}{$O(d)$~\cite{BackMC}} \\ \cline{3-5}
        & & \cboth & \cno & \cboth & & & \\ \cline{3-5}
        & & \cboth & \cboth & \cno & & & \\ \cline{3-6}
        & & \cyes & \cyes & \cyes & $\Theta(\min\{d, \highlight{n^{1/2}}\}$  & & \\ \hline
    \end{tabular}

    \end{adjustbox}}
    \caption{Overview of results.
    In the Case column, we indicate whether the given bounds are for a worst-case target vertex or averaged over all $n$ possible target vertices.
    In the Model column, circles indicate presence or absence of operations.
    The letters J, S, and A are abbreviations of $\Jump$, $\NeighSorted$, and $\Adj$, respectively.
    A full circle \cyes{} indicates that the operation is present in the model, and an empty circle \cno{} indicates that the operation is absent in the model.
    A half-full circle \cboth{} acts as a wildcard, indicating that the bounds hold both when the operation is present and absent.
    The results marked with $\star$ refer to folklore results and are formally proven in~\cite{DirectedPPR}.
    The terms highlighted in \hcolor indicate the new upper bounds.
    Note that all model combinations are covered, and the bounds are tight (up to polylogarithmic factors) for all possible values of $n$, $m$, and $\delta$, with $d=m/n$.
    }
    \label{tab:results}
\end{table}

\paragraph{Lower bounds.}
We construct tight lower bounds for all problems and model variants.
Before this work, only basic folklore lower bounds existed~\cite{DirectedPPR}, together with a tight lower bound for the single node problem~\cite{BackMC} in the bare adjacency-list model.
Our lower bounds are inspired by the directed lower bound constructions of~\cite{DirectedPPR}, which in turn are inspired by the directed single node construction of~\cite{STOC-BIPPR}.
As their constructions are asymmetric, they do not directly carry over to undirected graphs.
For each Personalized PageRank problem, we give an undirected hard instance with tunable parameters that can be set to match the different model variants.
This construction is more complex than in~\cite{DirectedPPR}, since the undirected setting requires additional parameters.

\paragraph{Upper bounds.}
We provide four main new upper bounds.
Let us briefly present the novel ideas behind each of them.
Some of them build on existing algorithms, which we modify or combine to get tight bounds.
\Cref{tab:algorithms} summarizes prior upper bounds together with our new ones.

\emph{Single node algorithm [\Cref{thm:sn-wc}].}
The previous best single node algorithm is known as \emph{backwards monte carlo}~\cite{BackMC} and runs in $O(m^{1/2})$ time. It works by generating $\alpha$-discounted random walks from $t$ and then using the reversibility property~\labelcref{eq:reverse-intro} to get an estimate for $\pi(t)$.
This is tight in all models except when we have access to all queries $\Jump$, $\NeighSorted$, and $\Adj$, in which case, we improve the running time to $O(n^{1/2})$.
This is the first application of this query combination in undirected graphs.
The main observation is that we can treat high-degree and low-degree neighbors of $t$ (separated efficiently using $\NeighSorted$) differently.
For low-degree neighbors we can efficiently compute their contribution to $\pi(t)$ using backwards monte carlo.
For high-degree neighbors, backwards sampling becomes too costly, and instead we estimate their contribution to $\pi(t)$ by moving to a random vertex using $\Jump$ and walking from there.
The intuition is that high-degree vertices have a higher probability of being hit by a random walk.
To check if we actually hit a high-degree neighbor, we use the $\Adj$ query.

\begin{table}[H]
  \centering
  \renewcommand{\arraystretch}{1.4}
  \begin{tabular}{|l|c|c|c|c|c|c|}
    \hline
    \multirow{2}{*}{\textbf{Method}} & \multicolumn{3}{c|}{\textbf{Model}} & \multirow{2}{*}{Prob.} & \multirow{2}{*}{WC} & \multirow{2}{*}{AC}  \\ \cline{2-4}
                                     & \textbf{J} & \textbf{S} & \textbf{A} & & &\\ \hline\hline
    
    Power Method \cite{page1999pagerank} &\cno &\cno &\cno & SS/ST & $\tilde{O}(m)$ & $\tilde{O}(m)$ \\\hline\hline
    
    Monte Carlo \cite{MC} &\cno &\cno &\cno & SS & $O(1/\delta)$ & $O(1/\delta)$ \\\hline\hline
    
    Backwards Push \cite{Lofgren_backpush} &\cno &\cno &\cno & ST & $O(n/\delta)$ & $O(d/\delta)$ \\\hline
    RBS \cite{RBS} &\cno &\cyes &\cno & ST & $\tilde O(n /\delta)$ & $\tilde O(1/\delta)$  \\\hline
    BiPPR ST \cite{DirectedPPR} &\cyes &\cno &\cno & ST & - & $O((m/\delta)^{1/2})$ \\\hline
    \Cref{thm:bi-st}  &\cyes &\cno &\cno & ST & $O(n/\delta^{1/2})$ & $O((m/\delta)^{1/2})$ \\\hline
    \Cref{thm:st-wc} &\cno &\cyes &\cno & ST & $\tilde O(n/\delta^{1/2})$ & $\tilde O((m/\delta)^{1/2})$ \\\hline
    \Cref{thm:back-improved} &\cno &\cno &\cno & ST & $O(n/\delta)$ & $O((1/\delta)^2+d)$ \\\hline
    \Cref{thm:back-improved} &\cno &\cyes &\cno & ST & $O(n/\delta)$ & $O((1/\delta)^2)$ \\\hline
    \Cref{thm:bi-st-avg} &\cyes &\cno &\cno & ST & - & $O((n/\delta)^{2/3}+d)$ \\\hline\hline

    BiPPR SP \cite{UndirectedBiPPR} &\cno &\cno &\cno & SP & $O((n/\delta)^{1/2})$ & $O((d/\delta)^{1/2})$ \\\hline
    \Cref{thm:back-improved-sp} &\cno &\cno &\cno & SP & - & $O((1/\delta)^{2/3}+d)$ \\\hline
    \Cref{thm:back-improved-sp} &\cno &\cyes &\cyes & SP & - & $O((1/\delta)^{2/3})$ \\\hline
    BJTWY \cite{DirectedPPR} &\cno &\cyes &\cyes & SP & - & $\tilde O((1/\delta)^{2/3})$ \\\hline\hline

    Backwards Monte Carlo \cite{BackMC} &\cno &\cno &\cno & SN & $O(m^{1/2})$ & $O(d)$ \\\hline
    \Cref{thm:sn-wc} &\cyes &\cyes &\cyes & SN & $O(n^{1/2})$ & $O(\min\{d,n^{1/2}\})$ \\\hline

  \end{tabular}
  \caption{
Overview of algorithms for (personalized) PageRank estimation.
The {Prob.} column specifies the problem variant (SS = single source, ST = single target, SP = single pair, SN = single node).
Columns {WC} and {AC} state worst-case and average-case expected time complexity.
For the meaning of the access models {J}, {S}, {A} and the circle notation, see \Cref{tab:results}.
  }
    \label{tab:algorithms}
\end{table}

\emph{Single target worst-case algorithm [\Cref{thm:st-wc}].}
The previous best worst-case algorithm for the single target problem is known as \emph{backwards push}~\cite{Lofgren_backpush} and runs in time $O(n/\delta)$.
It works by placing probability mass $1$ at $t$ and then repeatedly pushing it out into the graph, with each node $s$ maintaining an underestimate $p(u)$ approaching $\pi(u,t)$ as we distribute the probability mass more and more.
This algorithm is optimal if we disallow $\Jump$ and $\NeighSorted$.
When $\NeighSorted$ is allowed, we give an algorithm with running time $O(n/\delta^{1/2})$, using a novel idea of combining deterministic backwards push with the randomized backwards push of~\cite{RBS}, which uses $\NeighSorted$ to push to neighbors with probability inversely proportional to their degrees.
We find that it is more efficient to perform an initial phase of deterministic pushes, followed by a phase of randomized pushes.

\emph{Single target average-case algorithm [\Cref{thm:back-improved}].}
For the average-case single target problem, we provide an algorithm with running time $O(1/\delta^2+d)$.
The novel idea behind the $O(1/\delta^2)$ algorithm is to ignore all walks to $t$ that in the step before $t$ are at a neighbor $u$ with degree at least $\Omega(1/\delta)$.
We show that discarding these walks introduces an additive error of at most $O(\delta)$.
Backwards push~\cite{Lofgren_backpush} then runs in time $O(h_t/\delta + d(t))$ where $h_t$ is the number of neighbors of $t$ with degree at most $O(1/\delta)$.
Globally, $\sum_t h_t\leq O(n/\delta)$ since each node can contribute at most $O(1/\delta)$ to this sum, giving an average running time of $O((1/\delta)^2+d)$.
The $d$ term comes from scanning all neighbors of $t$ to find out which neighbors to ignore.
We note that this algorithm is deterministic, making it suitable for combination with randomized methods in a single pair setting.

\emph{Single pair average-case algorithm [\Cref{thm:back-improved-sp}].}
Here, we provide an $O((1/\delta)^{2/3}+d)$ algorithm in all models.
For this algorithm we use techniques similar to the bidirectional estimators of e.g.~\cite{UndirectedBiPPR}, combining our new hybrid backwards push from $t$ with monte carlo sampling from $s$.
Balancing these two parts gives the stated running time.
We further show how to remove the $d$ term if we allow $\NeighSorted$ and $\Adj$.
In this model our algorithm avoids the logarithmic factors present in the directed bound of~\cite{DirectedPPR}.

\paragraph{Estimation requirements.}
Similarly to~\cite{DirectedPPR}, we say that an algorithm solves the \emph{single pair problem} if for every undirected graph $G=(V,E)$, source $s\in V$, target $t\in V$, and approximation threshold $\delta\in (0,1]$, the algorithm outputs an estimate $\hat \pi(s,t)$ of $\pi(s,t)$ such that
\begin{align}
    \P{|\pih(s,t) - \pi(s,t)| \geq c \max\{\pi(s,t),\delta\}} \leq p_f \label{eq:estimation-requirement}
\end{align}
where $c$ and $p_f$ are small constants.

We say that an algorithm solves the \emph{single source problem} if, for every $G=(V,E)$, every source $s\in V$, and every $\delta\in(0,1]$, it outputs an an estimate $\pih(s,t)$ satisfying \Cref{eq:estimation-requirement} for every target $t\in V$.
The \emph{single target problem} is defined analogously, with the roles of source and target swapped.
Finally, we say that an algorithm solves the \emph{single node problem} if, for every $G=(V,E)$ and node $t\in V$, it outputs an an estimate $\pih(t)$ satisfying $\P{|\pih(t) - \pi(t)| \geq c \pi(t)} \leq p_f$.

\paragraph{Paper orginization.}
In~\Cref{sec:preliminaries} we cover preliminaries. In \Cref{sec:upper} we present upper bounds, and in \Cref{sec:lower} we present lower bounds.

\section{Preliminaries}\label{sec:preliminaries}

We denote the underlying (simple) undirected graph by $G=(V,E)$, where $n=|V|$ and $m=|E|$. 
For a vertex $v\in V$, we let $d(v)$ denote its degree, $\mc N(v)$ denote its set of neighbors, and $\mc N[v] = \mc N(v) \cup\{u\}$ denote the closed neighborhood. 
We write $d=m/n$ for the average degree. 
We use $\delta\in(0,1]$ for the approximation threshold of our estimation requirements, and we use $\tilde{O}(\cdot)$ notation to hide polylogarithmic factors in $n$ and $\delta$.
For any integer $k$, we denote by $[k]$ the set of integers $\{1,\dots,k\}$.

\subsection{Pagerank and PPR}

An $\alpha$-discounted walk starts at a vertex and, at each step, terminates with probability $\alpha$ and otherwise moves to a uniformly random neighbor.
The PageRank score of a vertex $v$ in an undirected graph $G$, denoted $\pi_G(v)$, is the probability that an $\alpha$-discounted random walk starting at a uniform random vertex ends in $v$.
The personalized PageRank (PPR) score of a vertex $v$ with respect to $u$ in $G$, denoted $\pi_G(u,v)$, is the probability that an $\alpha$-discounted random walk starting at $u$ ends in $v$.
When $G$ is clear from context, we drop the subscript and write $\pi(u,v)$ (or $\pi(v)$) instead of $\pi_G(u,v)$ (or $\pi_G(v)$).
Observe that $\pi(v) = \frac{1}{n}\sum_{u\in V} \pi(u,v)$.
By definition the following equations holds for each $v\in V$:
\begin{align}
    \pi(v) &= \sum_{u\in \mc N(v)} \frac{(1-\alpha)\pi(u)}{d(u)} + \frac{\alpha}{n},\label{eq:rec-node}\\
    \pi(u,v) &= \sum_{w\in \mc N(v)} \frac{(1-\alpha)\pi(u,w)}{d(w)} + \alpha \Ind{u=v}.\label{eq:rec-target}
\end{align}
Since we are considering an undirected graph $G$, the following reversibility property holds \cite{UndirectedBiPPR}.
\begin{lemma}\label{lem:reverse}
    For any vertices $u,v\in V$, we have $\pi(u,v)d(u) = \pi(v,u)d(v)$.
\end{lemma}
\begin{proof}
    We denote by $\mc P_{u,v}$ the set of all paths in $G$ from $u$ to $v$.
    For a path $p=v_1,v_2,\dots,v_k$ we write $\bar p = v_k,v_{k-1},\dots,v_1$ for the reverse path.
    We let $\pi(p)$ denote the probability that an $\alpha$-discounted random walk starting at $v_1$ follows the path $p$ and terminates at $v_k$.
    Then $\pi(p)=\frac{1}{d(v_1)},\dots \frac{1}{d(v_{k-1})} (1-\alpha)^{k-1}\alpha$, and we have $\pi(p)d(v_1)=\pi(\bar p) d(v_{k})$.
    From this, it follows that
    \begin{align*}
        \pi(u,t)d(u)
        = \sum_{p\in\mc P_{u,v}} \pi(p)d(u) 
        = \sum_{p\in \mc P_{u,v}} \pi(\bar p)d(v) 
         = \sum_{p\in \mc P_{v,u}} \pi(p)d(v)
         = \pi(v,u)d(v).
    \end{align*}
\end{proof}

\subsection{Monte Carlo Sampling}

Arguably, the simplest way of solving the single source problem is by generating multiple independent $\alpha$-discounted random walks starting at the source vertex, and estimate $\pi(s,u)$ for each $u\in V$ by the fraction of walks that end at $u$.
This is known as Monte Carlo sampling~\cite{MC}.
A small relative error is required only when $\pi(s,u)>\delta$, so standard Chernoff bounds imply that $O(1/\delta)$ walks are sufficient to achieve a constant relative error with constant success probability.
Since the expected length of an $\alpha$-discounted walk is $1/\alpha = O(1)$, this method solves the single source problem in time $O(1/\delta)$. 
In \Cref{sec:app-pseudocode}, we give an algorithm for sampling an $\alpha$-discounted walk from a vertex $u$, $\RW(G,u,\alpha)$, in expected time $O(1/\alpha)$.

In a similar way, we can also solve the single target problem by generating multiple independent $\alpha$-discounted random walks starting at the target vertex (known as \textit{backwards Monte Carlo sampling}) \cite{BackMC}.
This gives estimates $\pih(t,u)$ of $\pi(t, u)$ for every $u\in V$. Using the \textit{reversability property} stated in \Cref{lem:reverse}, this also gives estimates $\pih(u,t) = \pih(t,u) d(u)/d(t)$ of $\pi(u,t)$. 
A small relative error is required only when $\pi(u,t)>\delta$, or equivalently when $\pi(t,u)> \delta d(u)/d(t) = \delta'$.
Standard Chernoff bounds imply that $O(1/\delta')=O(d(t)/\delta)$ walks are sufficient to achieve a constant relative error with constant success probability. 
The worst-case complexity is $O(d(t)/\delta)=O(n/\delta)$, and averaged over all targets $t$, the complexity is $O(d/\delta)$.

\subsection{Local Push}\label{sec:local-push}

Another paradigm for computing PageRank starts with all probability mass at the target and then uses \Cref{eq:rec-target} to \emph{push}, or distribute, probability mass to neighboring vertices.
This algorithm is known as \texttt{BackwardsPush} \cite{Lofgren_backpush}, and the pseudocode is given in \Cref{alg:bacwards-push}.

\texttt{BacwardsPush} maintains two values, $p(u)$ and $r(u)$, for each vertex $u$.
The \emph{reserve} $p(u)$ serves as an underestimate of $\pi(u,t)$, and will approach $\pi(u,t)$ as we keep pushing probability mass around in the graph.
The \emph{residual} $r(u)$ is the remaining probability that currently sits at $u$ and still needs to be propagated according to \Cref{eq:rec-target}. 
When pushing at a vertex $v$, we increment the residual $r(u)$ of each neighbor $u$ by $(1-\alpha)r(v)/d(u)$, increment the reserve $p(v)$ by $\alpha r(v)$ and set $r(v)=0$.
This procedure maintains, for each $u \in V$, the invariant
\begin{align}
    \pi(u,t)= p(u)+\sum_{w\in V} r(w) \pi(u,w). \label{eq:backpush-invariant}
\end{align}
\begin{algorithm}[H]
    \caption{$\BackPush(v)$ }
    \label{alg:bacwards-push}
    \textbf{Input:} graph $G$, target vertex $t$, residual threshold $r_{\max}$, decay factor $\alpha$ \\
    \textbf{Output:} reserves $p()$ and residuals $r()$
    \begin{algorithmic}[1]
    \State $r(), p()\gets$ \text{empty dictionary with default value $0$}
    \State $r(t) \gets 1$
    \While{ exists $v$ with $r(v) > r_{\max}$}
        \State $\reserve(v)\gets \reserve(v) + \alpha r(v)$
        \For{$i$ from $1$ to $\Deg(v)$}
            \State $u \gets \Neigh(v,i)$
            \State $\residue(u) \gets \residue(u) + (1-\alpha)r(v) / \Deg(u)$
        \EndFor
        \State $\residue(v)\gets 0$
    \EndWhile
    \State \Return $\residue()$ and $\reserve()$
    \end{algorithmic}
\end{algorithm}

The algorithm continues pushing until $r(u)\leq r_{\max}$ for every vertex $u$, where $r_{\max}\in[0,1)$ is a fixed threshold. Upon termination, each reserve $p(u)$ will be within an additive error $r_{\max}$ of $\pi(u,t)$, since
$\pi(u,t)\geq p(u) = \pi(u,t)-\sum_{w\in V} r(w) \pi(u,w) \geq \pi(u,t)-r_{\max}\sum_{w\in V} \pi(u,w) = \pi(u,t)-r_{\max}$.
Hence, setting $r_{\max}\leq c\delta=O(\delta)$, \texttt{BackwardsPush} solves the single target problem.
In \cite{Lofgren_backpush} it is shown that the running time is bounded by $O(\sum_{u\in V}\frac{\pi(u,t) d(u)}{r_{\max}})$.
Using the reversibility property, \Cref{lem:reverse}, we can rewrite this as follows
\begin{align*}
    O\p{\sum_{u\in V}\frac{\pi(u,t) d(u)}{r_{\max}}} = 
    O\p{\frac{d(t)}{r_{\max}}\sum_{u\in V} \pi(t,u)} = 
    O\p{\frac{d(t)}{r_{\max}}}.
\end{align*}
Averaged over all targets $t$, the running time is $O(d/r_{\max})$.

\subsection{Power Method}

The local pushing algorithm $\BackPush$~\cite{Lofgren_backpush}, introduced above, works by pushing only from vertices $v$ whose residual exceeds the threshold $r_{\max}$.
If we instead run $\BackPush$ in a global manner, where we push at every vertex synchronously over $L$ rounds, we get the standard \emph{power method} for PageRank~\cite{page1999pagerank}.
Each round takes $O(m)$ time, since we push across every edge twice.

Let $r_i(v)$ be the residual at vertex $u$ after $i$ rounds. 
In one round, the residual value is scaled by a factor of $(1-\alpha)$, implying that $r_i(u)\leq (1-\alpha)^i$ for every $u\in V$.
Hence, after $L=\log_{1-\alpha} c\delta =O(\log 1/\delta)$ rounds, we have $r_L(u)\leq c\delta$ for every $u\in V$. 
By the same argument as in \Cref{sec:local-push}, this suffices to solve the single target problem.
The total running time is $O(mL)=\tilde O(m)$.

Similarly, one can run a forward variant that starts with all probability mass at the source and applies the same synchronous updates, thereby solving the single source problem in time $\tilde O(m)$.

\subsection{Randomized local push}\label{sec:rbs}
In \cite{RBS}, they give a randomized version of $\BackPush$ known as $\RBS$, utilizing the 
$\allowbreak \NeighSorted$ query.
The pseudocode is given in \Cref{alg:rbs}.
The algorithm runs in $L$ rounds, which is normally set to $\Theta(\log 1/\delta)$, and in each round, we push at each $v$ with $r_i(v)>0$.
The pushing is done deterministically (as in $\BackPush$) to a neighbor $u$ if the increment $\Delta_{i+1}(u,v)=(1-\alpha)r_i(u)/d(u)$ to $r_i(u)$ exceeds a predefined threshold $\theta\in(0,1)$.
Otherwise, it increases $r_i(u)$ by $\theta$ with probability $\Delta_{i+1}(u,v)$, and thus in expectation increases $r_i(u)$ with $\Delta_{i+1}(u,v)$.

Making a random decision for each $u \in \mc N(v)$ takes $\Omega(d(u))$ time.
Instead, since the increment to $r_i(u)$ is inversely proportional to $d(u)$, we can sample a random threshold $\tau$ uniformly from $[0, \theta]$ and use the $\NeighSorted$ query to increment each $r_i(u)$ by $\max\{\Delta_{i+1}(u,v), \theta\}$ if $\Delta_{i+1}(u,v) \geq \tau$. 
This still results in an expected increase of $\Delta_{i+1}(u,v)$.
Using this approach, it can be shown that the expected time complexity is at most $\tilde O(n\pi(t)/\theta)$ and setting $\theta=O(\delta)$ solves the single target problem in expected time $\tilde O(n\pi(t)/\delta)$.

\begin{algorithm}[H]\label{alg:rbs}
    \caption{$\RBS(v)$ }
    \textbf{Input:} reserves $\pp()$ and residuals $\rp()$ \\
    \textbf{Output:} 
    \begin{algorithmic}[1]
    \State $\hat r_i(), \hat p_i()\gets$ \text{empty dictionary with default value $0$, for each $i\in\{0,\dots,L\}$.}    
    \For{$i=0,1,2,\dots,L-1$}
        \For{each $v\in V$ with $r_i(v)>0$}
            \For{each $u\in \mc N(v)$}
            \State $\Delta_{i+1}(u,v)\gets \frac{(1-\alpha)\rh_i(v)}{\Deg(u)}$.
                \If{$\Delta_{i+1}(u,v)\geq \theta$}
                    \State $\rh_{i+1}(u)\gets \rh_{i+1}(u)+\Delta_{i+1}(u,v)$.
                \Else
                    \State $\rh_{i+1}(u)\gets \rh_{i+1}(u)+ \theta$ with probability $\Delta_{i+1}(u,v)/\theta$.
                \EndIf
            \EndFor
            \State $\ph_i(v)\gets \alpha \rh_i(v)$.
        \EndFor
    \EndFor
    
    \State \Return $\ph(u)=\sum_i \ph_i(u)$ for each $u$.
    \end{algorithmic}
\end{algorithm}

\section{Upper bounds}\label{sec:upper}

In this section, we present our new upper bounds, namely the results highlighted in \hcolor in \Cref{tab:results}.
In \Cref{sec:sn} we present an algorithm for the single node problem using $\NeighSorted, \Jump$, and $\Adj$, with time complexity $O(\min\{d(t),n^{1/2}\})$ and average time complexity $O(\min\{d,n^{1/2}\})$. 
In \Cref{sec:st-wc} we present an algorithm for the single target problem using $\NeighSorted$, with time complexity $O((d(t)n/\delta)^{1/2})$.
In \Cref{sec:st-avg} we present an algorithm for the single target problem with average time complexity $O((1/\delta)^2+d)$.
In \Cref{sec:sp} we present an algorithm for the single pair problem, with average time complexity $O((1/\delta)^{2/3}+d)$, which can be improved to $O((1/\delta)^{2/3})$ using $\NeighSorted$ and $\Adj$.
In \Cref{sec:up-jump}, we present the remaining two upper bounds, namely two algorithms for the single target problem with access to $\Jump$, one with time complexity $O((d(t)n/\delta)^{1/2})$ and the other with average time complexity $O((n/\delta)^{2/3}+d)$.

Combining our new upper bounds with the existing ones gives all upper bounds stated in \Cref{tab:results}.
For example, combining the $\tilde O(m)$ time complexity achieved by the $\pw$, the $O(d/\delta)$ time complexity achieved by $\BackPush$, and the $O((1/\delta)^2+d)$ time complexity from \Cref{sec:st-avg}, we get obtain an average time complexity of $\tilde O(\min\{m, d/\delta, (1/\delta)+d\})$.

\subsection{Single node}\label{sec:sn}

In this section, we present an algorithm solving the single node problem that achieves the following upper bound in the adjacency-list model with queries $\Jump,~ \NeighSorted$, and $\Adj$.
Note that the $d(t)$ of this upper bound does not appear in~\Cref{tab:results}, since our lower bounds do not parameterize $d(t)$.

\begin{theorem}\label{thm:sn-wc}
    Consider the adjacency-list model with $\Jump$, $\NeighSorted$, and $\Adj$.
    There exists an algorithm solving the single node problem with target $t$ in expected time $O(\min\{d(t), n^{1/2}\}=O(n^{1/2})$ and, when averaging over all targets $t$, expected time $O(\min\{d, n^{1/2}\})$.
\end{theorem}

\paragraph{Overview.}

By \Cref{eq:rec-node}, we can estimate $\pi(t)$ by estimating $\pi(u)$ for each neighbor $u$.
The algorithm treats low- and high-degree neighbors of $t$ differently.
Specifically, we fix a threshold $1\leq \tau\leq n$ and partition the neighbors into two sets $X_L$ containing all neighbors with degree at most $\tau$ and $X_H$ containing the remaining neighbors (degree larger than $\tau$).
The algorithm estimates the contribution of $X_L$ to $\pi(t)$ denoted $\pi_L(t)=\sum_{u\in X_L} {(1-\alpha)\pi(u)}/{d(u)}$, the contribution of $X_H$ to $\pi(t)$, denoted $\pi_H(t)=\sum_{u\in X_H} {(1-\alpha)\pi(u)}/{d(u)}$, and then combines this into an estimate $\hat\pi(t) = \pi_L(t)+\pi_H(t)+\frac{\alpha}{n}$ of $\pi(t)$.

To estimate $\pi_L(t)$, we first use $\NeighSorted$ to identify the set $X_L$ without scanning all $d(t)$ neighbors.
Because neighbors are sorted by degree, we can find $X_L$ in time $O(\log d(t))$.
We then generate $\alpha$-discounted random walks, each starting from a random vertex in $X_L$, and use the resulting samples to estimate $\pi_L(t)$ using the reversibility property of~\Cref{eq:reverse-intro}.
Using the fact that the degree of each vertex in $X_L$ is at most $\tau$, we show that it suffices to generate $O(\min\{d(t), \tau\})$ random walks to estimate $\pi_L(t)$ within a constant relative error.
Without any bound on the degree, we would get the same time complexity as in \cite{BackMC} of $O(\min\{d(t), m^{1/2}\})$.

To estimate $\pi_H(t)$ it becomes expensive to generate enough random walks from $X_H$, so we instead switch to forward sampling, which is more efficient as high-degree vertices have a higher probability of getting hit by a random walk starting from a uniformly random source.
Specifically, the algorithm repeatedly uses $\Jump$ to pick a uniform source vertex, generates an $\alpha$-discounted random walk, and takes its endpoint $x$. 
We then check whether $x\in X_H$ by using the $\Adj$ query and checking that $d(x)>\tau$.
If $x\in X_H$, the sample is included in the estimate of $\pi_H(t)$.
Since all vertices in $X_H$ have degree larger than $\tau$, we can show that $O(n/\tau)$ samples suffice to obtain a constant relative error of $\pi_H(t)$. 

The total running time is $O(\min\{d(t),\tau\} + \log d(t) + n/\tau)$.
Then, if $d(t) \leq n^{1/2}$, we set $\tau=n$, giving a running time of $O(d(t))=O(\min\{d(t), n^{1/2}\})$ and if $d(t)>n^{1/2}$, we set $\tau=n^{1/2}$, giving a running time of $O(n^{1/2}) = O(\min\{d(t), n^{1/2}\})$.

We will now show this formally.

\paragraph{Algortihm.}
The pseudocode is given in \Cref{alg:single-node}.
Initially, we set the estimate to $\alpha/n$, and identify the size of $X_L$.
Then we sample $w_L$ $\alpha$-discounted random walks from a random neighbor of $X_L$.
If a walk terminates at $u$, we increment $\pih(t)$ with $\frac{(1-\alpha)|X_L|}{n w_L d(u)}$, a quantity justified below.
Afterwards, we generate $w_H$ $\alpha$-discounted random walks, jumping to a uniformly random source vertex. 
If a walk terminates at $x\in X_H$, then we increment $\pih(t)$ with $\frac{1-\alpha}{w_H d(x)}$.
Note that line~\ref{alg-line:np-rw-4} corresponds to checking whether $x\in X_H$.

\begin{algorithm}[H]
\caption{$\SingleNode(G,t,\alpha,w_L,w_H,\tau)$ }
\label{alg:single-node}
\textbf{Input:} Undirected graph $G=(V,E)$, target $t\in V$, stopping probability $\alpha$, number of random walks $w_L$ and $w_H$, degree threshold $\tau$.\\
    \textbf{Output:} $\hat \pi(t)$ as an estimate of $\pi(t)$.
\begin{algorithmic}[1]
    \State $\pih(t) \leftarrow \alpha/n$.
    \State $|X_L|\leftarrow \max\{j\in [\Deg(t)] \mid \Deg(\Neigh(j))\leq \tau\}$ /\!/ {\color{gray} Binary search using $\NeighSorted$.}
    \For{$i\in [w_L]$}
        \State $x \gets \Neigh(\texttt{rand}(|X_L|))$ {\color{gray} /\!/ Random neighbor of $X_L$.}\label{alg-line:np-rw-1}
        \State $u \gets  \RW(G, x, \alpha)$ {\color{gray} /\!/ See \Cref{sec:app-pseudocode}}.\label{alg-line:np-rw-2}
        \State $\pih(t) \gets \pih(t) + \frac{1-\alpha}{n\cdot w_L}\cdot\frac{|X_L|}{\Deg(u)} $
    \EndFor
    \For{$i\in [w_H]$}
        \State $u \gets \Jump()$
        \State $x \leftarrow  \RW(G, u, \alpha)$ \label{alg-line:np-rw-3}
        \If{$x\in N(t) ~\land~ \Deg(x)> \tau$} \label{alg-line:np-rw-4}
            \State $\pih(t) \gets \pih(t) + \frac{1-\alpha}{w_H \cdot \Deg(x)}$
        \EndIf
    \EndFor
    \State \textbf{return} $p()$
\end{algorithmic}
\end{algorithm}

\paragraph{Analysis.}
This section presents the proof of \Cref{thm:sn-wc}.

Let $u_1,\dots,u_{w_L}$ and $x_1,\dots,x_{w_L}$ be the result of $\RW$ on line~\ref{alg-line:np-rw-2} and line~\ref{alg-line:np-rw-3} respectively of \Cref{alg:single-node}. 
Define the indicator variable $\chi_{i,u} = \Ind{u_i = u}$ for each $u\in V$ and $i\in [w_L]$, and the indicator variable $\phi_{i,x} = \Ind{x_i = x}$ for each $x\in X_H$ and $i\in [w_H]$.
It is easy to see that we can express $\pih(t)$ as a sum of these indicator variables
\begin{align*}
    \pih(t) 
    = \frac{(1-\alpha) \cdot |X_L|}{n \cdot w_L}\sum_{i\in [w_L]} \sum_{u\in V} \frac{\chi_{i,u}}{d(u)}
    + \frac{1-\alpha}{w_H}\sum_{i\in [w_H]} \sum_{x\in X_H} \frac{\phi_{i,x}}{d(x)}
    + \frac{\alpha}{n}.
\end{align*}
We will start by showing that $\pih(t)$ is an unbiased estimator.

\begin{lemma}\label{lem:sn-unbiased}
    $\E{\pih(t)}=\pi(t)$.
\end{lemma}
\begin{proof}
    In each of the $w_L$ samples, the algorithm first picks a start vertex $x$ uniformly at random from $X_L$, and then samples an $\alpha$-discounted walk from $x$.
    Therefore, for any $u \in V$
    \begin{align*}
        \E{\chi_{i,u}} 
        = \P{u_i = u} 
        = \frac{1}{|X_L|}\sum_{x\in X_L} \pi(x,u) 
        = \frac{1}{|X_L|}\sum_{x\in X_L} \pi(u,x) \frac{d(u)}{d(x)}.
    \end{align*}
    The last equality follows from \Cref{lem:reverse}.
    In each of the $w_H$ samples, the algorithm generates a random walk from a uniformly random start vertex. 
    For any $x\in X_H$, this sampled random walk terminates at $x$ with probability $\pi(x)$, and hence $\E{\phi_{i,x}}=\pi(x)$.

    Using the above, it follows by linearity of expectation that,
    \begin{align*}
        \E{\pih(t)} 
        &= \frac{1-\alpha}{n} \sum_{u\in V}\sum_{x\in X_L} \frac{\pi(u,x)}{d(x)}
        + (1-\alpha)\sum_{x\in X_H} \frac{\pi(x)}{d(x)}
        + \frac{\alpha}{n}\\
        &= \sum_{x\in \mc N(y)} \frac{(1-\alpha)\pi(x)}{d(x)}  + \frac{\alpha}{n}
    \end{align*}
    Then using \Cref{eq:rec-node} it follows that $\E{\pih(t)}=\pi(t)$.
\end{proof}

Next, the following lemma upper bounds the variance of the estimator $\pih(t)$ returned by the algorithm.
\begin{lemma}\label{lem:sn-var}
    $\Var{\pih(t)}\leq \p{\frac{|X_L|}{n w_L} 
        + \frac{1}{w_H \tau}}(1-\alpha) \pi(t)$
\end{lemma}
\begin{proof}
    Since the two groups of random walks are generated independently, we have that the $\chi$'s and $\phi$'s are independent of each other. This implies that
    \begin{align*}
        \Var{\pih(t) }
        &= \p{\frac{(1-\alpha) |X_L|}{n w_L}}^2\Var{\sum_{i\in [w_L]} \sum_{u\in V} \frac{\chi_{i,u}}{d(u)}}
        + \p{\frac{1-\alpha}{w_H}}^2\Var{\sum_{i\in [w_H]} \sum_{x\in X_H} \frac{\phi_{i,x}}{d(x)}}.
    \end{align*}
    Considering a pair of two different indicator variables $(\chi_{i,u},\chi_{i',u'})$, if $i\neq i'$ they are independent, and otherwise they are negatively correlated, hence $\Cov{\chi_{i,u},\chi_{i',u'}}\leq0$.
    This also holds for the $\phi$'s, and implies that
    \begin{align*}
        \Var{\pih(t) } 
        \leq \p{\frac{(1-\alpha) |X_L|}{n w_L}}^2 \sum_{i\in [w_L]} \sum_{u\in V} \frac{\Var{\chi_{i,u}}}{d(u)^2}
        + \p{\frac{1-\alpha}{w_H}}^2\sum_{i\in [w_H]} \sum_{x\in X_H} \frac{\Var{\phi_{i,x}}}{d(x)^2}.
    \end{align*}
    Each of the $\chi$'s and $\phi$'s are Bernoulli variables and therefore $\Var{\chi_{i,u}}\leq \E{\chi_{i,u}}$ and $\Var{\phi_{i,x}}\leq \E{{\phi_{i,x}}}$.
    Using the expression of $\E{\chi_{i,u}}$ and $\E{\phi_{i,x}}$ found in the proof of \Cref{lem:sn-unbiased}, we further get
    \begin{align*}
        \Var{\pih(t)}
        &\leq \frac{(1-\alpha)^2|X_L|}{n^2 w_L} \sum_{u\in V}\sum_{x\in X_L} \frac{\pi(u,x)}{d(x)d(u)}
        + \frac{(1-\alpha)^2}{w_H}\sum_{x\in X_H} \frac{\pi(x)}{d(x)^2}\\
        &\leq \frac{(1-\alpha)|X_L|}{n w_L} \sum_{x\in X_L} \frac{(1-\alpha)\pi(x)}{d(x)}
        + \frac{(1-\alpha)}{w_H}\sum_{x\in X_H} \frac{(1-\alpha)\pi(x)}{d(x)^2}
    \end{align*}
    Finally, recall that $d(x)>\tau$ for each $x\in X_H$, so using \Cref{eq:rec-node}, it follows that
    \begin{align*}
        \Var{\pih(t)} &\leq \p{\frac{|X_L|}{n w_L} + \frac{1}{w_H \tau}}(1-\alpha) \pi(t),
    \end{align*}
    which finishes the proof.
\end{proof}
We are now ready to prove the theorem.
\begin{proof}[Proof of \Cref{thm:sn-wc}]
    We will set $w_L, w_H$ and $\tau$ so that $\Var{\pih(t)}\leq (c\pi(t))^2 p_f$. Then, since by \Cref{lem:sn-unbiased} the estimator is unbiased, we can use Chebyshev's inequality, and we get that $\P{|\pih(t)-\pi(t)|\geq c\pi(t)}\leq p_f$, hence solving the single node problem.

    By \Cref{lem:sn-var}, it suffices to set $w_L\geq\frac{2(1-\alpha)|X_L|}{n c^2 \pi(t) p_f}$ and $w_H\geq\frac{2(1-\alpha)}{\tau c^2 \pi(t) p_f}$.
    However, we do not know the value of $\pi(t)$ at the beginning of the algorithm. 
    To alleviate this, we will derive a lower bound of $\pi(t)$, so we need not know it. 
    We have
    \begin{align*}
         \pi(t) \geq \sum_{x\in X_L} \frac{(1-\alpha) \pi(x)}{d(x)} + \frac{\alpha}{n} \geq |X_L| \frac{(1-\alpha)\alpha}{\tau n} + \frac{\alpha}{n} 
         \geq \max\curly{|X_L| \frac{(1-\alpha)\alpha}{\tau n}, \frac{\alpha}{n}},
    \end{align*}
    where the first inequality comes from \Cref{eq:rec-node}.
    Using this lower bound and that $|X_L|\leq d(t)$, we set 
    \begin{align*}
        w_L = \min\curly{\frac{2\tau}{c^2\alpha p_f}, \frac{2(1-\alpha) d(t)}{c^2\alpha p_f}} 
        = O\p{ \min\curly{\tau, d(t)}},
    \end{align*}
    and
    \begin{align*}
        w_H = \frac{2(1-\alpha) n}{\tau c^2 \alpha p_f} = O\p{\frac{n}{\tau}},
    \end{align*}
    and note that these satisfy the wished inequalities.

    Now we will argue how to set $\tau$ to get the claimed running time.
    First of all, to find $|X_L|$ by using the $\NeighSorted$, we can find $|X_L|$ in $O(\log d(t))$ using binary search.
    Then, combining this with the sampling process, the total running time is
    \begin{align*}
        O(\log d(t) + w_H + w_L) = O\p{\log d(t) + \min\{\tau, d(t)\} + \frac{n}{\tau}}.
    \end{align*}
    Clearly, setting $\tau=\Theta(n^{1/2})$ suffices to achieve a worst-case running time of $O(n^{1/2})$.
    However, we can sometimes get a better running time by doing the following.
    If $d(t)\leq n^{1/2}$ then we set $\tau=n$, giving a running time of $O(d(t)) = O(\min\{d(t),n^{1/2}\})$.
    Otherwise, we set $\tau = n^{1/2}$ as before, giving a running time of $O(\sqrt{n}) = O(\min\{d(t),n^{1/2}\})$.
    Finally, averaging over all targets, this running time is $O(\min\{d,n^{1/2}\})$.
\end{proof}

\subsection{Single target worst-case}\label{sec:st-wc}

In this section, we will present an algorithm solving the single target problem, achieving the following upper bound in the model with $\NeighSorted$.

\begin{theorem}\label{thm:st-wc}
    Consider the adjacency-list model with $\NeighSorted$.
    There exists an algorithm solving the single target problem with target $t$ in expected time $O((d(t) n/\delta)^{1/2})=O(n/\delta^{1/2})$ and, when averaging over all targets $t$, expected time $O((m/\delta)^{1/2})$.
\end{theorem}

\paragraph{Overview.}
The algorithm will use a modified version of $\RBS$ \cite{RBS} described in \Cref{sec:rbs}.
In each push of $\RBS$ at a vertex $u$, we increase the estimate $\ph(u)$ by at least $\theta$, where $\theta$ is the push threshold.
Since $\ph(u)$ is always an underestimate of $\pi(u,t)$ in expectation, the expected number of pushes is $O\p{\sum_u \frac{\pi(u,t)}{\theta}}=O(n \pi(t)/\theta)$, which also bounds the expected running time due to the randomized pushing.
Setting $\theta=O(\delta)$ solves the single target problem, as mentioned in \Cref{sec:rbs}.
As we can have $n\pi(t)=\Omega(d(t))$, this gives an expected running time of $O(d(t)/\delta)$, which is not the running time that we opt for.
Moreover, we cannot set $\theta = o(\delta)$, as this would cause the additive error to be larger than $\delta$.

Instead, we begin with a phase where we push deterministically as long as there exists a vertex $u$ with a large residual $r(u)$, that is, when $r(u)> \tau$, where $\tau\in (0,1)$ is a predefined threshold.
This will accelerate the running time, as we will quickly distribute large chunks of probability mass, instead of performing many small pushes.
Once this phase terminates, we will perform the remaining pushes using the standard randomized pushing with push threshold $\delta$.

Considering the running time, the deterministic phase costs at most $O(d(t)/\tau)$, which is the worst-case running time of $\BackPush$ \cite{Lofgren_backpush}.
After this, we know that for each $u$ the current estimate satisfies $\pi(u,t)-\ph(u) \leq \tau$.
It then follows that, in expectation, $\RBS$ will perform at most $\tau/\delta$ pushes at each vertex, resulting in a running time of $O(n\tau/\delta)$.
Balancing both phases, by setting $\tau=(d(t)\delta/n)^{1/2}$, results in an expected running time of $O((d(t)n/\delta)^{1/2})$.

\paragraph{Algorithm.} 
The pseudocode is given in \Cref{alg:st-wc}. 
We start with $\BackPush$ with $r_{\max}$, which returns all the reserve and residual values.
We then run $\RBS$ with a push threshold of $\theta$, and it will run for $L$ rounds.
These parameters will be set in later proofs.

\begin{algorithm}[H]
\caption{}
\label{alg:st-wc}
\textbf{Input:} Undirected graph $G=(V,E)$, target vertex $t\in V$, stopping probability $\alpha$, residual threshold $r_{\max}$, number of pushing rounds $L$, push threshold $\theta$. \\
\textbf{Output:} Dictionary for estimates $\hat p()$ .
\begin{algorithmic}[1]
    \State $p(),r() \leftarrow \BackPush(t, r_{\max})$
    \State $\hat p() \leftarrow \RBS(t, L, \theta)$ initialized with $\hat p_0(), \hat r_0() \leftarrow p(), r()$
    \State \textbf{return} $\hat p()$
\end{algorithmic}
\end{algorithm}

\paragraph{Analysis.}
We let $p_i(u)$ and $r_i(u)$ be the deterministic version of $\ph_i(u)$ and $\rh_i(u)$ of the $\RBS$ algorithm, that is, the reserve and residual we get if we push $\Delta_{i+1}(u,v)$ to every neighbor $u$.
To show that our estimate satisfies the requirement of \Cref{eq:estimation-requirement}, we want to show that the actual reserve and residual values, i.e. $\ph_i(u)$ and $\rh_i$, do not deviate too far from the ``ideal'' deterministic version.
This basically follows from the below lemma, which is shown in \cite{RBS}.
\begin{lemma}\label{lem:rbs-prop}
  The following statements hold for the $\RBS$ algorithm,
  \begin{itemize}
    \item[(1)] $\E{\ph_i(u)} = p_i(u)$ for each $i\in\{0,\dots,L\}$,
    \item[(2)] $\E{\rh_i(u)} = r_i(u)$ for each $i\in\{0,\dots,L\}$, and
    \item[(3)] $\Var{\ph(u)} \leq L \theta p(u)$.
  \end{itemize}
\end{lemma}

Furthermore, in the normal $\RBS$ all $r_0(u)$ and $p_0(u)$ is set to $0$ except $r_0(t)$, so it is no longer clear what $p_i(u)$ represents.
Normally, $p_i(u)=\pi^i(u,t)$ where $\pi^i(u,t)$ is the probability than an $\alpha$-discounted walk starting at $u$ terminates at $t$ after exactly $i$ steps. 
However, we know that the deterministic reserves and residuals correspond to $\BackPush$, and using the invariant given in \Cref{eq:backpush-invariant}, we get
\begin{align*}
    \pi(u,t) &= p(u) + \sum_{v\in V} \pi(u,v) r_L(v)
\end{align*}
Using this, we start by showing that setting $L=\log_{1-\alpha}(s)$ rounds only introduces a $s$ additive error.

\begin{lemma}\label{lem:set-L}
    For any $s>0$, if $L=\log_{1-\alpha}(s)$ then $\pi(u,t)-s\leq p(u)\leq \pi(u,t)$
\end{lemma}
\begin{proof}
    We start by showing that $r_i(v) \leq (1-\alpha)^i$ for all $v$ using induction over $i$.
    For $i=0$, $r_i(v)\leq 1$ is clearly true. 
    Consider any vertex $v$.
    It receives residual mass from each neighbor in the $(i-1)$'th round, and we have
    \begin{align*}
        r_i(v) = \sum_{w\in \mc N(v)} \frac{(1-\alpha) r_{i-1}(w)}{d(v)} 
        \leq (1-\alpha)^i \sum_{w\in \mc N(v)} \frac{1}{d(v)} = (1-\alpha)^i.
    \end{align*}
    After $L$ steps we have $r_L(v) \leq (1-\alpha)^L = s$, implying that
    \begin{align*}
        \pi(u,t) &= p(u) + \sum_{v\in V} \pi(u,v) r_L(v) \leq p(u) + s.
    \end{align*}
    Thus, $\pi(u,t)-s\leq p(u)\leq \pi(u,t)$.
\end{proof}

We are now ready to show \Cref{thm:sn-wc}.
\begin{proof}[Proof of \Cref{thm:st-wc}]
    We will show that setting 
    \begin{align*}
        \theta=c^2\delta p_f/(4L), ~\text{and}~ L=\log_{1-\alpha}(c\delta/2)
    \end{align*}
    suffices for the algorithm to solve the single target problem.
    We will set $r_{\max}$ later to achieve the stated running time.
    By \Cref{lem:set-L}, we have $\pi(u,t)-c\delta/2\leq p(u)\leq \pi(u,t)$ for each $u\in V$.
    Then by the triangle inequality $|\ph(u) - \pi(u,t)|\leq |\ph(u) - p(u)| + c\delta/2$.
    Therefore,
    \begin{align*}
        p_u = \P{|\ph(u)- \pi(u,t)| \geq c \max\{\delta, \pi(u,t)\}} &\leq \P{|\ph(u)- p(u)| + c\delta /2 \geq c \max\{\delta, \pi(u,t)\}}\\
        &\leq \P{|\ph(u)- p(u)| \geq (c/2) \max\{\delta, \pi(u,t)\}}.
    \end{align*}
    Recall that from \Cref{lem:rbs-prop} we have that $\Var{\ph(u)} \leq \theta L p(u)$.
    Then Chevyshev's inequality, \Cref{lem:rbs-prop} and $p(u)\leq \pi(u,t)$ implies that
    \begin{align*}
        p_u \leq \frac{4 \theta L p(u)}{c^2 \max\{\delta, \pi(u,t)\}^2} 
        \leq \frac{\delta p(u) p_f}{\max\{\delta, \pi(u,t)\}^2}
        \leq p_f.
    \end{align*}
    This finishes the first part of the proof.

    Note that the $\BackPush$ phase runs in time $O(d(t)/r_{\max})$.
    For the remaining part of the proof, we will try to upper bound the running time of the randomized pushing phase.
    As noted in \cite{RBS}, conditioned on $\{\rh_i\}$ the push operation at $v$ in round $i$ can be implemented in expected time $O\p{\sum_{u\in \mc N(v)} \Delta_{i+1}(u,v) /\theta + 1}$.
    Thus, we can upper bound the expected time complexity of the randomized pushing phase as follows, where $N$ is the expected number of times we have a non-zero residual,
    \begin{align}
        O\p{\frac{1}{\theta}\sum_{i=0}^{L-1} \sum_{v\in V} \sum_{u\in \mc N(v)} \E{\Delta_{i+1}(v,u)} + N}. \label{eq:st-rbs-rt}
    \end{align}
    Note that $\E{\rh_{i+1}(v) \mid \{\rh_i\}} = \sum_{u \in \mc N(v)} \Delta_{i+1}(v,u)$, as $v$ receives an expected residual increment from each neighbor $u$ of $\Delta_{i+1}(v,u)$, implying together with \Cref{lem:rbs-prop} that $r_{i+1}(v) = \E{\rh_{i+1}(v)} = \sum_{u \in \mc N(v)} \E{\Delta_{i+1}(v,u)}$. 
    Then for the first term of \Cref{eq:st-rbs-rt} we get
    \begin{align}
        \frac{1}{\theta}\sum_{i=0}^{L-1} \sum_{v\in V} \sum_{u\in \mc N(v)} \E{\Delta_{i+1}(v,u)} 
        = \frac{1}{\theta}\sum_{i=0}^{L-1} \sum_{v\in V} r_{i+1}(v)
        = \frac{1}{\theta \alpha}\sum_{v\in V} (p(v) - p_0(v)). \label{eq:st-rbs-rt2}
    \end{align}
    It follows from the invariant given in \Cref{eq:backpush-invariant} maintained by $\BackPush$, that for each $u\in V$ we have $\pi(u,t) \geq p_0(u) \geq \pi(u,t) - r_{\max}$, and hence
    \begin{align*}
         p(v) - p_0(v) \leq p(v) - \pi(u,t) + r_{\max} \leq r_{\max}.
    \end{align*}
    Applying this to \Cref{eq:st-rbs-rt2} shows that the first term of \Cref{eq:st-rbs-rt} is upper bounded by $O(n r_{\max} / (\theta \alpha)) = \tilde O(n r_{\max} / \delta)$.
    Now we are only missing to bound $N$.
    For $i>0$ we only push at $v$ when $\rh_i(v)\geq \theta$. 
    When $i=0$, there are at most $O(d(t)/r_{\max})$ non-zero residuals, as the running time of $\BackPush$ upper bounds the number of pushes and hence number of non-zero residuals. 
    Thus,
    \begin{align*}
        N \leq O\p{\frac{d(t)}{r_{\max}}} + \E{\sum_{i=0}^{L-1} \sum_{v\in V} \frac{\rh_{i+1}(v)}{\theta}}  =O\p{\frac{d(t)}{r_{\max}}}   + \frac{1}{\theta} \sum_{i=0}^{L-1} \sum_{v\in V} r_{i+1}(v)
        = \tilde O\p{\frac{d(t)}{r_{\max}} + \frac{n r_{\max}}{\delta}}. 
    \end{align*}
    The second-to-last equality follows from \Cref{lem:rbs-prop} and the last equality follows from above (note that the sum-term appears in \Cref{eq:st-rbs-rt2}).
    Thus, the final expected time complexity is $O\p{\frac{d(t)}{r_{\max}} + \frac{n r_{\max}}{\delta}}$, and setting $r_{\max}=(d(t)\delta/ n)^{1/2}$, yields a time complexity of $O((n d(t)/\delta)^{1/2})$, finishing the proof.
\end{proof}

\subsection{Single target average-case}\label{sec:st-avg}

In this section, we will describe an algorithm $\BackPushNew$ with the following properties

\begin{theorem}\label{thm:back-improved}
    Consider the adjacency-list model.
    There exists an algorithm solving the single target problem with target $t$ in expected time $O(d(t)/\delta)$ and, when averaging over all targets $t$, expected time $O(\min\{d/\delta,(1/\delta)^2 +d\})$.
    
    Considering instead the adjacency-list model with $\NeighSorted$, there exists an algorithm solving the single target problem in expected time $O(\min\{d/\delta,\allowbreak (1/\delta)^2\})$, averaged over all targets.
\end{theorem}

\paragraph{Overview.}
The algorithm is a modified version of the standard $\BackPush$ \cite{Lofgren_backpush}.
Instead of starting with all probability mass at the target $t$, we will distribute the probability mass among the neighbors of $t$, but ignoring all neighbors with degree larger than $O(1/\delta)$.
In that way, on average, we avoid a lot of pushes on these expensive high-degree vertices, reducing the asymptotic time complexity.
We show that this only introduces an additive error of at most $O(\delta)$.
Specifically, we show that the worst-case running time is $O(h_t/\delta+d(t))$ where $h_t$ is the number of neighbors of $t$ with degree at most $O(1/\delta)$.
The $d(t)$ comes from the initial scan of neighbors to distribute the probability mass.
Importantly, $\sum_t h_t\leq O(n/\delta)$ since each vertex can contribute at most $O(1/\delta)$ to this sum, giving an average running time of $O((1/\delta)^2+d)$.
Note that the worst-case time complexity is the same as $\BackPush$, as a target $t$ can have $h_t=d(t)$, resulting in a time complexity of $O(d(t)/\delta)$.
Furthermore, if the $\NeighSorted$ is available, then we can distribute the initial probability mass in time $h_t$, and we can remove the $d$ term from the average time complexity.

\paragraph{Algorithm.}
The pseudocode of $\BackPushNew$ is given in \Cref{alg:bacwards-push-improved}.
The input is a target vertex $t\in V$ and a residual threshold $r_{\max}$, which will be set in later proofs.
Let $Y\subseteq \mc N(t)$ be the subset of neighbors of $t$ with degree at most $1/r_{\max}$.
Furthermore, let $X= \mc N(t)\setminus Y$.
Then for each $x\in X$ we have $d(x)>1/r_{\max}$.
We will now run the $\BackPush$ algorithm, but with the initialization slightly changed.
Instead of setting $r(t)=1$ and all other values to zero, we set $p(t)=\alpha$ and $r(y)=(1-\alpha)/d(y)$ for each $y\in Y$ and all other values is set to zero.
We then run $\BackPush$ with push threshold $r_{\max}$.

\paragraph{Analysis.}
Below, we let $p(u)$ and $r(u)$ denote the final reserve and residual value.
We will start by showing that the following invariant is maintained, where the proof is given in \Cref{sec:missing-proofs}.

\begin{lemma}\label{lem:bpnew-invariant}
    For the target vertex $t$, the push operation in $\BackPushNew$ algorithm maintain the following invariant for each $u\in t$:
    \begin{align*}
        \pi(u,t)=p(u) + \sum_{v \in V}\p{r(v)+\frac{1-\alpha}{d(v)}\Ind{v\in X}} \pi(u,v) . 
    \end{align*}
\end{lemma}

We have by the following Lemma that the final estimate $p(u)$ is within a $O(r_{\max})$ additive error of $\pi(u,t)$.
\begin{lemma}\label{lem:bpnew-add}
    $\pi(u,t)-2r_{\max}\leq p(u)\leq \pi(u,t)$.
\end{lemma}
\begin{proof} 
    For each vertex $v\in V$ if $v\in X$ then $d(x)>1/r_{\max}$ and also $r(v)\leq r_{\max}$ implying that
    $
        r(v)+\frac{1-\alpha}{d(v)}\Ind{v\in X}
        \leq r_{\max} + (1-\alpha) r_{\max}
        \leq 2 r_{\max}
    $.
    Using this, we have that
    \begin{align*}
        p(u) &= \pi(u,t) - \sum_{v \in V}\p{r(v)+\frac{1-\alpha}{d(v)}\Ind{v\in X}} \pi(u,v)\\
        &\leq \pi(u,t) - 2r_{\max} \sum_{v \in V} \pi(u,v)\\
        &= \pi(u,t) - 2r_{\max}.
    \end{align*}
    
\end{proof}
We are now ready to prove the theorem.
\begin{proof}[Proof of \Cref{thm:back-improved}]
    By \Cref{lem:bpnew-add}, setting $r_{\max}=c\delta/2$, given an additive error of at most $c\delta$ therby solving the single target problem. 
    For the remaining part of this proof, we will argue about the running time.
    
    First of all, we can upper bound $p(u)$ as follows:
    \begin{align*}
        p(u) &= \pi(u,t) - \sum_{v \in V}\p{r(v)+\frac{1-\alpha}{d(v)}\Ind{v\in X}} \pi(u,v) \\
        &\leq \pi(u,t) - \sum_{x \in X}\frac{1-\alpha}{d(x)} \pi(u,x)\\
        &= \sum_{y\in Y} \frac{(1-\alpha)\pi(u,y)}{d(y)} + \alpha \Ind{u=t} && (\text{By \Cref{eq:rec-target}})\\
        &= \frac{(1-\alpha)}{d(u)}\sum_{y\in Y} \pi(y,u) + \alpha \Ind{u=t}. && (\text{By \Cref{lem:reverse}})
    \end{align*}
    For a vertex $u$, we denote by $R(u)=p(u)/\alpha - \Ind{u = t}$ the total amount of residual pushed from $u$ to its neighbors.
    Whenever we push at a vertex $u$, we have $r(u)>r_{\max}$, and after pushing, we set $r(u)=0$.
    Therefore the number of pushes at $u$ is upper bounded by $R(u)/r_{\max}$.
    We can then upper bound the running time as follows,
    \begin{align*}
        O\p{\sum_{u\in V}\frac{R(u)}{r_{\max}}d(u) }
        = O\p{\sum_{u\in V} \frac{p(u) - \alpha\Ind{u = t}}{\alpha r_{\max}}d(u)}
    \end{align*}
    Using the above upper bound of $p(u)$, we further get
    \begin{align*}
        O\p{\sum_{u\in V} \frac{p(u) - \alpha\Ind{u = t}}{\alpha r_{\max}}d(u)}
        = O\p{\frac{(1-\alpha)}{\alpha r_{\max}} \sum_{y\in Y} \sum_{u\in V} \pi(y,u)}
        = O\p{\frac{|Y|}{r_{\max}}}.
    \end{align*}
    Note that $O(|Y|/r_{\max}) = O(d(t)/r_{\max})$.
    
    For the average-case running time (when averaging over all targets), we combine two bounds. 
    First, the running time is at most $\frac{1}{n}\sum_{t\in V} O(d(t)/\delta) = O(d/\delta)$. 
    Secondly, we bound the running time by applying a charging argument. 
    For the $\Theta(|Y|/\delta)$ random walks sampled at $t$, each vertex $y\in Y$ is charged a cost of $\Theta(1/\delta)$. A vertex $u$ is charged at most $d(u)$ times, and only when $d(u)\leq \Theta(1/\delta)$, so its total charge is at most $O((1/\delta)^2)$. 
    Thus, the average-case running time is $O(d+(1/\delta)^2)$, by also adding the average-cost of the neighboring search.
    Combining the two bounds give $O(\min\{d/\delta,(1/\delta)^2+d\})$.
    Note that if the $\NeighSorted$ operation is available, then we can do the neighboring search in time $| Y|$, and then the average-case running time is simply $O(\min\{d/\delta,(1/\delta)^2)\}$.
\end{proof}

\subsubsection{Extending to single pair average case.}\label{sec:sp}
In this section, we will describe an algorithm $\BiPPRAvg$ with the following running time.

\begin{theorem}\label{thm:back-improved-sp}
    Consider the adjacency-list model.
    There exists an algorithm solving the single pair problem in expected time $O((1/\delta)^{2/3}+d)$, averaged over all targets.
    
    Considering instead the adjacency-list model with $\NeighSorted$ and $\Adj$, there exists an algorithm solving the single pair problem in expected time $O((1/\delta)^{2/3})$, averaged over all targets.
\end{theorem}

We note that in \cite{DirectedPPR}, they also give an $O((1/\delta)^{2/3})$ algorithm in the model with $\NeighSorted$ and $\Adj$.
However, our algorithm achieves the corresponding bound with a substantially simpler design and analysis.
Moreover, we show that for undirected graphs, the same running time can be obtained even without these additional queries, up to an additive $d$ term

\paragraph{Overview.}
The bidirectional estimator given in \cite{UndirectedBiPPR} combines the backwards-push algorithm with Monte Carlo sampling.
Using the same idea for our new $\BackPushNew$ algorithm gives the stated running times.

\paragraph{Algorithm.}
$\BiPPRAvg$ first runs $\BackPushNew$ with target $t$ and push threshold $r_{\max}$ computing reserves $p()$ and residuals $r()$. 
$\BiPPRAvg$ then modifies all residuals for each $x\in X$ by updating $r(x)\gets r(x) + (1-\alpha)/d(v)$.
Then it follows from \Cref{lem:bpnew-invariant} that for each $u\in V$:
\begin{align}
    \pi(u,t) = p(u) + \sum_{v\in V}r(v)\pi(u,v).\label{eq:bippr-1}
\end{align}
From \Cref{eq:bippr-1}, we use Monte Carlo sampling to construct an estimator $\pih(s,t)$ of $\pi(s,t)$.
Specifically, we will generate $n_r$ $\alpha$-discounted random walks starting at $s$ using the $\RW$ subroutine, and let $u_i$ be the end vertex of the $i$ random walk for each $i\in [n_r]$.
We then construct the bidirectional estimator as follows:
\begin{align}
    \pih(s,t) = p(s) + \frac{1}{n_r} \sum_{i\in [n_r]} \sum_{v\in V} r(v) \Ind{v=u_i}. \label{eq:bippr-2}
\end{align}
The pseudocode is given in \Cref{sec:app-pseudocode}.

\paragraph{Analysis.}

We start by bounding the residual values as follows.
\begin{lemma}\label{lem:sp-residual}
    $r(v) < 2\cdot r_{\max}$ for each $v\in V$.
\end{lemma}
\begin{proof}
    After running $\BackPush$ we have that $r(v)\leq r_{\max}$ for each $v\in V$.
    Then for each $x\in X$, we add $(1-\alpha)/d(x) < (1-\alpha)r_{\max} \leq r_{\max}$, and hence the residual value is less than $ 2\cdot r_{\max}$ at each vertex.
\end{proof}

This is an unbiased estimator with variance bounded as follows.
\begin{lemma}\label{lem:sp-unbiased}
    $\E{\pih(s,t)} = \pi(s, t)$ and $\Var{\pih(s,t)}\leq 2 r_{\max} \pi(s,t)/n_r$.
\end{lemma}
\begin{proof}
    Since $\E{\Ind{v=u_i}} = \pi(s,v)$, we have that 
    \begin{align*}
        \E{\pih(s,t)} = p(s) + \sum_{v\in V} r(v) \pi(u,v) = \pi(s,t),
    \end{align*}
    where the last equality follows from \Cref{eq:bippr-1}.
    The variance can be bounded as follows,
    \begin{align*}
        \Var{\pih(s,t)} 
        &\leq \frac{1}{n_r^2} \sum_{i\in [n_r]} \sum_{v\in V} r(v)^2 \cdot \Var{\Ind{v=u_i}}
        \leq \frac{1}{n_r}\sum_{v\in V} r(v)^2 \pi(s,v).
    \end{align*}
    Using \Cref{lem:sp-residual}, we further get
    \begin{align*}
        \Var{\pih(s,t)} 
        &\leq \frac{2 r_{\max}}{n_r}\sum_{v\in V} r(v) \pi(s,v) \leq \frac{2 r_{\max}}{n_r} \pi(s,t).
    \end{align*}
\end{proof}

\begin{proof}[Proof of \Cref{thm:back-improved-sp}]
    From \Cref{lem:sp-unbiased}, setting $n_r=2r_{\max}/(c^2\delta p_f) = O(r_{\max}/\delta)$, results in a variance upper bounded by $\Var{\pih(s,t)} \leq c^2 \delta \pi(s,t) p_f$.
    Then, by Chebyshev's inequality, we get that
    \begin{align*}
        \P{|\pih(s,t)-\pi(s,t)| \geq c\max\{\pi(s,t),\delta\}} \leq \frac{c^2 \delta \pi(s,t) p_f}{c^2\max\{\pi(s,t),\delta\}^2} \leq p_f,
    \end{align*}
    thereby solving the single pair problem.
    
    By \Cref{thm:back-improved}, the average-case running time of $\BackPushNew$ is $O((1/r_{\max})^2+d)$.
    Correcting the residuals can be done in $O(d(t))$ time or $O(d)$ averaged over all targets, and the Monte Carlo part runs in $O(r_{\max}/\delta)$ time in expectation.
    We can balance the two terms $(1/r_{\max})^2$ and $r_{\max}/\delta$ by setting $r_{\max}=\delta^{1/3}$, resulting in a average running time of $O((1/\delta)^{2/3} + d)$.

    Lastly, we will describe how to get rid of this $d$ term in the running time if the $\NeighSorted$ and $\Adj$ queries are available.
    By \Cref{thm:back-improved}, the $\NeighSorted$ query removes the $d$ term from the running of the $\BackPushNew$ algorithm.
    To remove the $d$ term incurred by correcting the residuals so that \Cref{eq:bippr-1} is satisfied, we observe that we do not need to update the residuals for each $x\in X$ before the Monte Carlo phase.
    Instead, if we end up at vertex $u_i$, we can check whether $u_i\in X$ by using the $\Adj$ query to check if $u_i\in \mc N(t)$ and by checking that $d(u_i)>1/r_{\max}$.
    This gives the exact same estimate $\pih(s,t)$ for each $u$, but the average expected running time is improved to $O((1/\delta)^{2/3})$. 
\end{proof}

\subsection{Bidirectional estimators using Jump}\label{sec:up-jump}
In this section, we will present two algorithms, both solving the single target problem.
The first has running time $O(({n d(t)}/{\delta})^{1/2})$ in the worst-case, and the second has running time $O(({n }/{\delta})^{2/3}+d)$ averaged over all targets. 
The algorithms are based on the bidirectional approach similar to \cite{DirectedPPR, UndirectedBiPPR}.

\begin{theorem}\label{thm:bi-st}
    Consider the adjacency-list model with $\Jump$.
    There exists an algorithm solving the single target problem with target $t$ in expected time $O(({n d(t)}/{\delta})^{1/2})=O(n/\delta^{1/2})$ and, when averaging over all targets $t$, expected time $O((m/\delta)^{1/2})$.
\end{theorem}
\begin{proof}
    We use a bidirectional approach that combines Monte Carlo sampling with $\BackPush$ similar to \cite{UndirectedBiPPR}.
    Each Monte Carlo sample starts by drawing a uniformly random vertex using $\Jump$, and then simulating an $\alpha$-discounted walk from this vertex.
    As show in \cite{UndirectedBiPPR}, estimating $\pi(u,t)$ under the accuracy requirement of \Cref{eq:estimation-requirement} for a single pair $(u,t)$ can be done using $O(r_{\max}/\delta)$ random walks, along with the $\BackPush$ algorithm requiring a running time of $O(d(t)/r_{\max})$.
    Therefore, to solve the single target problem, the expected running time of the Monte Carlo part becomes $O(n r_{\max}/\delta)$ to estimate $\pi(u,t)$ for each $u\in V$.
    Balancing both terms, by setting $r_{\max}=(\delta d(t) /n)^{1/2}$, gives a time complexity of $O(({n d(t)}/{\delta})^{1/2})$.
\end{proof}

\begin{theorem}\label{thm:bi-st-avg}
    Consider the adjacency-list model with $\Jump$.
    There exists an algorithm solving the single target problem in expected time $O(({n }/{\delta})^{2/3}+d)$, averaged over all targets.
\end{theorem}
\begin{proof}
    We use a bidirectional approach that combines Monte Carlo sampling with $\BackPushNew$ similar to \cite{UndirectedBiPPR}.
    Each Monte Carlo sample starts by drawing a uniformly random vertex using $\Jump$, and then simulating an $\alpha$-discounted walk from this vertex.
    As show in \cite{UndirectedBiPPR}, estimating $\pi(u,t)$ under the accuracy requirement of \Cref{eq:estimation-requirement} for a single pair $(u,t)$ can be done using $O(r_{\max}/\delta)$ random walks, along with the $\BackPushNew$ algorithm requiring a running time of $O((1/r_{\max})^2+d)$ by \Cref{thm:back-improved}.
    Therefore, to solve the single target problem, the expected running time of the Monte Carlo part becomes $O(n r_{\max}/\delta)$ to estimate $\pi(u,t)$ for each $u\in V$.
    Balancing both terms, by setting $r_{\max}=(\delta /n)^{1/3}$, gives a time complexity of $O(({n }/{\delta})^{2/3}+d)$.
\end{proof}

\section{Lower bounds}\label{sec:lower}

In this section, we present our lower bounds.
See~\Cref{tab:results} for an overview.
The lower bounds are inspired by~\cite{DirectedPPR}, and are similarly based on a reduction of a decision problem to PageRank estimation.
We construct two graphs that only differ by a few edges, but on which $\pi(s,t)$ differs significantly.
Any algorithm that can estimate $\pi(s,t)$ within the required accuracy must therefore be able to distinguish between the two graphs.
The following lemma describes this reduction.

\begin{lemma}\label{sp-separation}
  Let $G$ be a graph containing vertices $s$ and $t$.
  Let $\hat\pi(s,t) \in [0, 1]$ be an estimate of $\pi(s,t)$ satisfying $\abs{\hat\pi(s,t)-\pi(s,t)} \leq c\max\{\pi(s,t),\delta\}$.
  Then,
  \begin{itemize}
    \item $\hat\pi(s,t) \leq c\delta$ if $\pi(s,t) = 0$, and
    \item $\hat\pi(s,t) > c\delta$ if $\pi(s,t) > 2c\delta$.
  \end{itemize}
\end{lemma}
\begin{proof}
  We assume $c < \frac12$.
  If $\pi(s,t)=0$, the claimed inequality follows by substituting $\pi(s,t)$ by zero in the assumed inequality.
  If $\pi(s,t) > 2c\delta$, we have $\pi(s,t)-\hat\pi(s,t) \leq \abs{\hat\pi(s,t)-\pi(s,t)} \leq c\max\{\pi(s,t),\delta\} < \frac12\pi(s,t)$, so $\hat\pi(s,t) > \frac12\pi(s,t) > c\delta$.
\end{proof}

To obtain our lower bounds from the above lemma, it suffices to construct graph classes that are hard to distinguish, but on which $\pi(s,t)$ differs significantly.
We will do this by hiding an important edge in a graph using a \emph{swap}, which only modifies the graph locally, but significantly alters the PPR.
Let us first define this swap.

\begin{definition}
  Let $G = (V, E)$ be a graph with a fixed order on the neighbors of each vertex.
  For each $q = (q_1,q_2,q_3,q_4) \in V^4$, write $E^-_q=\{\{q_1,q_2\}, \{q_3,q_4\}\}$, $E^+_q=\{\{q_1,q_3\},\{q_2,q_4\}\}$, and $E_q^\pm = E_q^-\cup E_q^+$.
  A quadruple $q \in V^4$ is \emph{swappable} if $E^-_q \subseteq E$ and $E^+_q \subseteq \binom{V}{2} \setminus E$.
  For a swappable quadruple $q \in V^4$, write $G_q$ for the graph with vertex set $V$ and edge set $(E \setminus E^-_q) \cup E^+_q$, inheriting the orders of vertex neighbors from $G$, placing edges of $E^+_q$ at indices where edges of $E^-_q$ were.
  For a set $Q \subseteq V^4$ of swappable quadruples, define $\mc G(G, Q) = \{G\} \cup \{G_q \mid q \in Q\}$.
\end{definition}

See~\Cref{fig:sp} for an illustration of a swap.
We now lower bound the number of queries required to detect whether a swap has occurred.
In the following lemma, the input is either $G$ or $G$ with a single quadruple in $Q$ swapped, and the algorithm must decide which is the case.

\begin{lemma}\label{swap-lb}
  Consider the adjacency-list model augmented with $U \subseteq \{\Jump,\allowbreak \NeighSorted,\allowbreak \Adj\}$.
  Let $G=(V,E)$ be a graph and $Q \subseteq V^4$ a set of swappable quadruples.
  If $\NeighSorted \in U$, assume that $d(q_1)=d(q_4)$ and $d(q_2)=d(q_3)$ for all $(q_1,q_2,q_3,q_4) \in Q$.

  Fix an ordered list $W$ of at least one vertex in $V$.
  If $\Jump \in U$, let $V_W = V$, and otherwise, let $V_W \subseteq V$ be the union of components in $G$ intersecting $W$.
  If $\Adj \in U$, let $E_W = \binom{V_W}{2}$, and otherwise, let $E_W = \binom{V_W}{2} \cap E$.

  Let $\mc A$ be an algorithm which, when given a graph $H \in \mc G(G, Q)$ and the list $W$ of some of its vertices, produces an output $\mc A(H, W)$.
  Assume that $\P{\mc A(G, W) = \mc A(G_q, W)} \leq 2p_f$ for all $q \in Q$.
  Then, the expected query complexity of $\mc A$ on $(G, W)$ is $\Omega(\abs{Q}/K)$, where $K = \max_{e \in E_W}\abs{\{q \in Q \mid e \in E^\pm_q\}}$.
\end{lemma}
\begin{proof}
  Let $S \subseteq \binom{V}{2}$ be the set of edges and other size-2 vertex sets queried by $\mc A$ when given $(G, W)$.
  Note that if $S \cap E_q^\pm = \emptyset$ for some $q \in Q$, then the query/answer sequence of $\mc A$ is identical on $(G, W)$ and $(G_q, W)$ (for the same random bits), hence $\mc A(G, W) = \mc A(G_q, W)$.
  To see that $\NeighSorted$ behaves the same in both graphs, note that the degree of no vertex changes when going from $G$ to $G_q$, and that when e.g. $q_1$ has its neighbor $q_2$ replaced by $q_3$, this new neighbor $q_3$ takes the place of $q_2$ in the degree-sorted adjacency list of $q_1$, assuming appropriate manual tiebreaking, since we assume $d(q_2)=d(q_3)$ when $\NeighSorted \in U$.
  Taking probabilities, $\P{S \cap E_q^\pm = \emptyset} \leq \P{\mc A(G, W) = \mc A(G_q, W)} \leq 2p_f$ for all $q \in Q$.

  Note that $S \subseteq E_W$, since the algorithm can only visit nodes in $V_W$ and hence only query pairs in $E_W$.
  We have $\abs{\{q \in Q \mid S \cap E^\pm_q \neq \emptyset\}} \leq \sum_{e \in S} \abs{\{q \in Q \mid e \in E^\pm_q\}} \leq K\abs{S}$, where the last inequality uses $S \subseteq E_W$.
  Taking expectations yields $\E{\abs{\{q \in Q \mid S \cap E^\pm_q \neq \emptyset\}}} \leq K\E{\abs{S}}$.
  On the other hand, since $\P{S \cap E_q^\pm = \emptyset} \leq 2p_f$ for all $q \in Q$, we have $\E{\abs{\{q \in Q \mid S \cap E^\pm_q \neq \emptyset\}}} = \sum_{q \in Q}\P{S \cap E^\pm_q \neq \emptyset} \geq (1-2p_f)\abs{Q}$.
  Combining these bounds gives $\E{\abs{S}} \geq (1-2p_f)\abs{Q}/K = \Omega(\abs{Q}/K)$.
\end{proof}

To apply~\Cref{swap-lb}, the idea is to construct a graph $G$ together with a large set of swappable quadruples $Q$, chosen so that many quadruples are edge-disjoint.
All of our lower bounds follow by combining \Cref{sp-separation} and \Cref{swap-lb}, together with an appropriate choice of a hard graph instance.

We emphasize that the lower bound holds on the graph $G$.
It is not neccessary to give a random graph, or a different graph depending on the algorithm, as is often the case with arguments based on Yao's principle.
In our model it is not even necessary to randomly permute the labels of $G$, since we only allow querying vertices given in the input (which is $W$ in~\Cref{swap-lb}) or returned by previous queries.

In the following subsections, we prove our lower bounds for each problem: single pair in \Cref{sec:sp-lb}, single source in \Cref{sec:ss-lb}, single target in \Cref{sec:st-lb}, and single node in \Cref{sec:sn-lb}.

\subsection{Single pair lower bounds}\label{sec:sp-lb}

We first combine \Cref{sp-separation} and \Cref{swap-lb} to obtain a general lemma for proving lower bounds for the single pair problem.

\begin{lemma}\label{sp-universal}
  Let $\mc A$ be an algorithm solving the single pair problem with approximation threshold $\delta$ in the adjacency-list model augmented with $U \subseteq \{\Jump, \NeighSorted, \Adj\}$.
  Let $G = (V, E)$ be a graph and $Q \subseteq V^4$ a set of swappable quadruples.
  If $\NeighSorted \in U$, assume that $d(q_1)=d(q_4)$ and $d(q_2)=d(q_3)$ for all $(q_1,q_2,q_3,q_4) \in Q$.
  Let $s$ and $t$ be vertices of $G$ with $\pi_G(s,t)=0$ and $\pi_{G_q}(s,t) > 2c\delta$ for all $q \in Q$.

  If $\Jump \in U$, let $V_W = V$, and otherwise, let $V_W \subseteq V$ be the union of components in $G$ containing $s$ or $t$.
  If $\Adj \in U$, let $E_W = \binom{V_W}{2}$, and otherwise, let $E_W = \binom{V_W}{2} \cap E$.

  Then, the expected query complexity of $\mc A$ on $G$ with source $s$ and target $t$ is $\Omega(\abs{Q}/K)$ where $K = \max_{e \in E_W}\abs{\{q \in Q \mid e \in E^\pm_q\}}$.
\end{lemma}
\begin{proof}
  The lower bound follows from~\Cref{swap-lb} with $W = (s,t)$ if we show $\P{\mc A(G, W) = \mc A(G_q, W)} \leq 2p_f$ for all $q \in Q$.
  Let $\hat\pi_G(s,t)$ be the estimate produced by $\mc A$ on $(G, W)$ and $\hat\pi_{G_q}(s,t)$ be the estimate produced by $\mc A$ on $(G_q, W)$.
  Then for each $q \in Q$, the event $\mc A(G, W) = \mc A(G_q, W)$ implies that $\hat\pi_G(s,t)=\hat\pi_{G_q}(s,t)$, so either $\hat\pi_G(s,t) > c\delta$ or $\hat\pi_{G_q}(s,t) \leq c\delta$.
  Both of the latter two events happen with probability at most $p_f$ by~\Cref{sp-separation}, since $\mc A$ solves the single pair problem.
  By a union bound, $\P{\mc A(G, W) = \mc A(G_q, W)} \leq 2p_f$.
\end{proof}

We now prove the worst-case lower bound for the single pair problem in the adjacency-list model with all query types available by applying \Cref{sp-universal} on an explicit hard instance.
The hard instance is the disjoint union of two complete bipartite graphs, and the swap replaces one edge in each component by two crossing edges.
For every swap, $\pi(s,t)$ changes from $0$ in $G$ to at least $2c\delta$ in the swapped graph, so distinguishing whether a swap occurred reduces to estimating $\pi(s,t)$ in the single pair problem.

\begin{theorem}\label{sp-wc-j-s-a}
  Consider the adjacency-list model with $\Jump$, $\NeighSorted$, and $\Adj$.
  For any $n$ and $m$ with $1 \leq n \leq m \leq n^2$ and any $\delta \in (0, 1]$, there exists a graph $G=(V,E)$ with $\Theta(n)$ vertices, $\Theta(m)$ edges, and vertices $s, t \in V$, such that for any algorithm solving the single pair problem, the expected running time on $G$ with source $s$, target $t$, and approximation threshold $\delta$ is $\Omega(\min\{m, (n/\delta)^{1/2}, 1/\delta\})$.
\end{theorem}
\begin{proof}
  We assume $c < \frac{1}{16}(1-\alpha)^6\alpha$.
  Let $x$ and $y$ be positive integer parameters to be chosen later.
  Define disjoint sets $A$, $B$, $C$, and $D$ with $\abs{A}=\abs{D}=x$ and $\abs{B}=\abs{C}=y$.
  For disjoint sets $X$ and $Y$, we denote by $K_{X,Y}$ the complete bipartite graph with vertex sex $X\cup Y$ and edge set $\{\{x,y\}\mid x \in X, y\in Y\}$.
  Construct $G=(V,E)$ as the union of the complete bipartite graphs $K_{A,B}$ and $K_{C,D}$.
  Add to the construction an independent subgraph with $n$ vertices and $m$ edges.
  Requiring $\max\{x, y\} \leq n$ and $xy \leq m$, we get a graph with $\Theta(n)$ vertices and $\Theta(m)$ edges.

  \begin{figure}[h]
    \centering
    \begin{tikzpicture}[
  styleN1/.style={circle, draw, inner sep=0, minimum size=18pt, line width=0.5pt},
  styleD1/.style={line width=0.5pt},
  styleD2/.style={line width=0.7pt, densely dashed, color=red},
  styleD3/.style={line width=0.7pt, color=blue},
  >=latex
  ]
  \pgfmathsetmacro{\K}{4};
  \pgfmathsetmacro{\L}{3};
  \pgfmathsetmacro{\labelOffset}{.3}
  \pgfmathsetmacro{\layerOne}{0}
  \pgfmathsetmacro{\layerTwo}{\layerOne-1}
  \pgfmathsetmacro{\layerThree}{\layerTwo-1}
  \pgfmathsetmacro{\layerFour}{\layerThree-1}   
  \pgfmathsetmacro{\swapA}{3};
  \pgfmathsetmacro{\swapB}{2};
  \pgfmathsetmacro{\swapC}{3};
  \pgfmathsetmacro{\swapD}{1};
  \foreach \i in {1,...,\K} {
    \pgfmathsetmacro{\xi}{\i-1-(\K-1) / 2}
    \node[styleN1] (vA\i) at (\xi, \layerOne) {\ifthenelse{\i = 2}{$s$}{}\ifthenelse{\i = \swapA}{$q_A$}{}};
    \node[styleN1] (vD\i) at (\xi, \layerFour) {\ifthenelse{\i = 3}{$t$}{}\ifthenelse{\i = \swapD}{$q_D$}{}};
  }
  \foreach \i in {1,...,\L} {
    \pgfmathsetmacro{\xi}{\i-1-(\L-1) / 2}
    \node[styleN1] (vB\i) at (\xi, \layerTwo) {\ifthenelse{\i = \swapB}{$q_B$}{}};
    \node[styleN1] (vC\i) at (\xi, \layerThree) {\ifthenelse{\i = \swapC}{$q_C$}{}};
  }
  \draw[styleD2] (vA\swapA) -- (vB\swapB);
  \draw[styleD2] (vC\swapC) -- (vD\swapD);
  \foreach \i in {1,...,\K} {
    \foreach \j in {1,...,\L} {
      \ifthenelse{\i = \swapA \AND \j = \swapB}{}{\draw[styleD1] (vA\i) -- (vB\j)};
      \ifthenelse{\i = \swapD \AND \j = \swapC}{}{\draw[styleD1] (vD\i) -- (vC\j)};
    }
  }
  \draw[styleD3] (vA\swapA) to[bend right=15] (vC\swapC);
  \draw[styleD3] (vB\swapB) to[bend left=15] (vD\swapD);
  \node (VA) at (\K/2+\labelOffset, \layerOne) {$A$};
  \node (VB) at (\K/2+\labelOffset, \layerTwo) {$B$};
  \node (VC) at (\K/2+\labelOffset, \layerThree) {$C$};
  \node (VD) at (\K/2+\labelOffset, \layerFour) {$D$};
\end{tikzpicture}
    \caption{Hard instance for detecting a swap for $x=4$, $y=3$, and some $q=(q_A,q_B,q_C,q_D)$.
    With the {\color{red} red} edge pair, $s$ does not reach $t$, but with the {\color{blue} blue} edge pair, $s$ reaches $t$.
    An algorithm has to distinguish between the two cases to satisfy the estimation requirements.
    }
    \label{fig:sp}
  \end{figure}
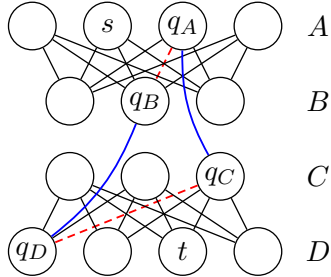

  Fix $s \in A$ and $t \in D$.
  Let $Q = A \times B \times C \times D$, and note that this is a set of swappable quadruples.
  We wish to apply~\Cref{sp-universal} with the above construction of $G$ and choice of $s$, $t$, and $Q$.
  See~\Cref{fig:sp} for an illustration of $G$ and $G_q$.
  Note that the assumption on the degrees of nodes in~\Cref{sp-universal} is satisfied.
  To obtain the lower bound of $\Omega(\abs{Q}/K)$ from~\Cref{sp-universal}, it remains to show, that $\pi_G(s,t)=0$ and $\pi_{G_q}(s,t) \geq 2c\delta$ for all $q \in Q$.
  It is clear that $\pi_G(s,t)=0$, since $s$ and $t$ are in different components of $G$.
  For $q=(q_A,q_B,q_C,q_D)\in Q$, a random walk in $G_q$ starting at $s$ moves from $s$ to $B\setminus\{q_B\}$ with probability $(1-\alpha)(1-1/y)$, then to $q_A$ with probability $(1-\alpha)/x$, then to $q_C$ with probability $(1-\alpha)/y$, then to $D\setminus\{q_D\}$ with probability $(1-\alpha)(1-1/x)$, then to $C\setminus\{q_C\}$ with probability $(1-\alpha)(1-1/y)$, then to $t$ with probability $(1-\alpha)/x$, and then stops with probability $\alpha$.\footnote{We consider this convoluted route, since we might have $q_A=s$ or $q_D=t$.}
  Taking the product, we get $\pi_{G_q}(s,t) \geq (1-\alpha)^6\alpha(1-1/x)(1-1/y)^2/(x^2y) \geq 2c\delta$, if we require that $\min\{x,y\} \geq 2$ and $x^2y\delta \leq 1$ when choosing $x$ and $y$.

  In our setting, $\abs{Q} = x^2y^2$.
  To determine $K = \max_{e \in \binom{V}{2}}\abs{\{q \in Q \mid e \in E_q^\pm\}}$, consider separately edges in $A\times B$, $C \times D$, $A \times C$, and $B \times D$.
  No vertex pair outside these four sets is in $E_q^\pm$ for any $q \in Q$.
  For every $e \in A \times B$ we have $\abs{\{q \in Q \mid e \in E_q^\pm\}} = \abs{C \times D} = xy$, and the same is true for the three remaining sets, so $K = xy$.
  We therefore get a lower bound of $\Omega(\abs{Q}/K) = \Omega(xy)$.

  We are now ready to choose $x$ and $y$, remembering to ensure $2 \leq \min\{x,y\}$, $\max\{x, y\} \leq 2n$, $xy \leq 4m$, and $x^2y\delta \leq 1$.

  \emph{Case 1:} For $0 < \delta \leq \frac{1}{md}$, set $x = \floor{2d}$ and $y = 2n$, giving a lower bound of $\Omega(m)$.

  \emph{Case 2:} For $\frac{1}{md} \leq \delta \leq \frac{1}{n}$, set $x = \floor{2(n\delta)^{-1/2}}$ and $y = 2n$, giving a lower bound of $\Omega((n/\delta)^{1/2})$.

  \emph{Case 3:} For $\frac{1}{n} \leq \delta \leq 1$, set $x = 2$ and $y = \floor{2/\delta}$, giving a lower bound of $\Omega(1/\delta)$.
\end{proof}

We next prove an average-case lower bound by taking many disjoint copies of the hard instance from \Cref{sp-wc-j-s-a}.
For a constant fraction of choices of $s$ and $t$, we can use the same swap argument, giving the following bound for the average over all source-target pairs.
\begin{theorem}\label{sp-ac-j-s-a}
  Consider the adjacency-list model with $\Jump$, $\NeighSorted$, and $\Adj$.
  For any $n$ and $m$ with $n \leq m \leq n^2$ and any $\delta \in (0, 1]$, there exists a graph $G=(V,E)$ with $\Theta(n)$ vertices and $\Theta(m)$ edges, such that for any algorithm solving the single pair problem, the expected running time on $G$ with approximation threshold $\delta$, averaging over all souces $s \in V$ and targets $t \in V$, is $\Omega(\min\{m, (d/\delta)^{1/2}, (1/\delta)^{2/3}\})$.
\end{theorem}
\begin{proof}
  We assume $c < \frac{1}{16}(1-\alpha)^4\alpha$.
  We will construct a graph containing multiple copies of the graph from the proof of~\Cref{sp-wc-j-s-a}.
  Let $x$ and $y$ be positive integer parameters to be chosen later.
  Define disjoint sets $A_i$, $B_i$, $C_i$, and $D_i$ with $\abs{A}=\abs{D}=x$ and $\abs{B}=\abs{C}=y$ for all $i \in \{1,\ldots, \floor{n/x}\}$.
  Construct $G=(V,E)$ as the union of the complete bipartite graphs $K_{A_i,B_i}$ and $K_{C_i,D_i}$ over all $i \in \{1,\ldots, \floor{n/x}\}$.
  Add to the construction an independt subgraph with $n$ vertices and $m$ edges.
  Requiring $y \leq x \leq 2n$ and $y \leq 2d$, we get a graph with $\Theta(n)$ vertices and $\Theta(m)$ edges.

  To obtain the lower bound on the average over all source-target pairs, it suffices to show it for $\Theta(n^2)$ fixed source-target pairs.
  So fix a source $s \in A_{i_s}$ and target $t \in D_{i_t}$ for some $i_s$ and $i_t$, noting that there are $\Theta(n^2)$ possible choices of such source-target pairs.
  Let $Q = A_{i_s} \times B_{i_s} \times C_{i_t} \times D_{i_t}$, and note that this is a set of swappable quadruples.
  We wish to apply~\Cref{sp-universal} with the above construction of $G$ and choice of $s$, $t$, and $Q$.
  For an illustration of $G$ and $G_q$ restricted to $A_{i_s} \cup B_{i_s} \cup C_{i_t} \cup D_{i_t}$,~\Cref{fig:sp} can be used again.
  Exactly as in the proof of~\Cref{sp-wc-j-s-a}, we get $\pi_G(s,t)=0$ and $\pi_{G_q}(s,t) \geq 2c\delta$ for all $q \in Q$, if we require $\min\{x,y\} \geq 2$ and $x^2y\delta \leq 1$, thus giving a lower bound of $\Omega(\abs{Q}/K) = \Omega(xy)$ by~\Cref{sp-universal}.

  We are now ready to choose $x$ and $y$, remembering to ensure $2 \leq y \leq x \leq 2n$, $y \leq 2d$, and $x^2y\delta \leq 1$.

  \emph{Case 1:} For $0 < \delta \leq \frac{1}{mn}$, set $x = 2n$ and $y = \floor{2d}$, giving a lower bound of $\Omega(m)$.

  \emph{Case 2:} For $\frac{1}{mn} \leq \delta \leq \frac{1}{d^3}$, set $x = \floor{2(d\delta)^{-1/2}}$ and $y = \floor{2d}$, giving a lower bound of $\Omega((d/\delta)^{1/2})$.

  \emph{Case 3:} For $\frac{1}{d^3} \leq \delta \leq 1$, set $x = y = \floor{2(1/\delta)^{1/3}}$, giving a lower bound of $\Omega((1/\delta)^{2/3})$.
\end{proof}

Consider the adjacency-list model with $\Jump$ and exactly one of $\NeighSorted$ and $\Adj$.
In this setting, it turns out that we always obtain a lower bound of $\Omega(\min\{1/\delta,d\})$.
The construction again uses disjoint complete bipartite graphs, but the swap replaces an edge incident to $s$ and an edge incident to $t$ so that $s$ and $t$ become adjacent.
Any algorithm must then inspect a constant fraction of the neighbors of $s$ or of $t$ to detect whether the swap occurred, which yields the claimed bound.

\begin{theorem}\label{sp-ac-j-s-xor-a}
  Consider the adjacency-list model with $\Jump$ and either $\NeighSorted$ or $\Adj$, but not both.
  For any $n$ and $m$ with $n \leq m \leq n^2$ and any $\delta \in (0, 1]$, there exists a graph $G=(V,E)$ with $\Theta(n)$ vertices and $\Theta(m)$ edges, such that for any algorithm solving the single pair problem, the expected running time on $G$ with approximation threshold $\delta$, averaging over all souces $s \in V$ and targets $t \in V$, is $\Omega(\min\{m, (d/\delta)^{1/2}, (1/\delta)^{2/3}+d, 1/\delta\})$.
\end{theorem}
\begin{proof}
  We assume again $c \leq \frac14(1-\alpha)^4\alpha$.
  Following~\Cref{sp-ac-j-s-a}, it suffices to show a lower bound of $\Omega(\min\{d, 1/\delta\})$ in this setting.

  Let us first consider the ajacency-list model with $\Jump$ and $\NeighSorted$.
  Take the construction of the graph $G$ from~\Cref{sp-ac-j-s-a}, putting $y=x$, so we only have one parameter $x$, to be chosen later.
  This construction requires $x \leq d$ to give a graph with $\Theta(n)$ vertices and $\Theta(m)$ edges.
  Again, it suffices to show the lower bound for $\Theta(n^2)$ source-vertex pairs, so fix $s \in A_{i_s}$ and $t \in D_{i_t}$ for some $i_s$ and $i_t$.
  Now choose $Q = \{s\} \times B \times \{t\} \times C$, noting that this is a set of swappable quadruples.
  See~\Cref{fig:sp-direct} for an illustration of $G$ and $G_q$ for some $q \in Q$.

  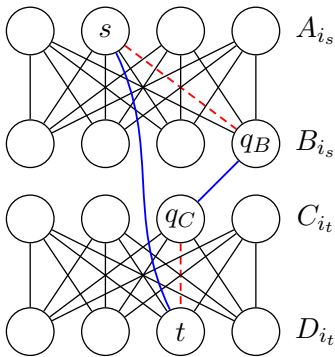
\begin{figure}[h]
    \centering
    \begin{tikzpicture}[
  styleN1/.style={circle, draw, inner sep=0, minimum size=18pt, line width=0.5pt},
  styleD1/.style={line width=0.5pt},
  styleD2/.style={line width=0.7pt, densely dashed, color=red},
  styleD3/.style={line width=0.7pt, color=blue},
  >=latex
  ]
  \pgfmathsetmacro{\K}{4};
  \pgfmathsetmacro{\L}{4};
  \pgfmathsetmacro{\labelOffset}{.3}
  \pgfmathsetmacro{\layerOne}{0}
  \pgfmathsetmacro{\layerTwo}{\layerOne-1.5}
  \pgfmathsetmacro{\layerThree}{\layerTwo-1}
  \pgfmathsetmacro{\layerFour}{\layerThree-1.5}
  \pgfmathsetmacro{\swapA}{2};
  \pgfmathsetmacro{\swapB}{4};
  \pgfmathsetmacro{\swapC}{3};
  \pgfmathsetmacro{\swapD}{3};
  \foreach \i in {1,...,\K} {
    \pgfmathsetmacro{\xi}{\i-1-(\K-1) / 2}
    \node[styleN1] (vA\i) at (\xi, \layerOne) {\ifthenelse{\i = \swapA}{$s$}{}};
    \node[styleN1] (vC\i) at (\xi, \layerThree) {\ifthenelse{\i = \swapC}{$q_C$}{}};
    \node[styleN1] (vD\i) at (\xi, \layerFour) {\ifthenelse{\i = \swapD}{$t$}{}};
  }
  \foreach \i in {1,...,\L} {
    \pgfmathsetmacro{\xi}{\i-1-(\L-1) / 2}
    \node[styleN1] (vB\i) at (\xi, \layerTwo) {\ifthenelse{\i = \swapB}{$q_B$}{}};
  }
  \draw[styleD2] (vA\swapA) -- (vB\swapB);
  \draw[styleD2] (vC\swapC) -- (vD\swapD);
  \foreach \i in {1,...,\K} {
    \foreach \j in {1,...,\L} {
      \ifthenelse{\i = \swapA \AND \j = \swapB}{}{\draw[styleD1] (vA\i) -- (vB\j)};
    }
  }
  \foreach \i in {1,...,\K} {
    \foreach \j in {1,...,\K} {
      \ifthenelse{\i = \swapD \AND \j = \swapC}{}{\draw[styleD1] (vD\i) -- (vC\j)};
    }
  }
  \draw[styleD3] (vA\swapA) to[out=-64,in=116] (vD\swapD);
  \draw[styleD3] (vB\swapB) to[bend left=0] (vC\swapC);
  \node (VA) at (\K/2+\labelOffset, \layerOne) {$A_{i_s}$};
  \node (VB) at (\K/2+\labelOffset, \layerTwo) {$B_{i_s}$};
  \node (VC) at (\K/2+\labelOffset, \layerThree) {$C_{i_t}$};
  \node (VD) at (\K/2+\labelOffset, \layerFour) {$D_{i_t}$};
\end{tikzpicture}
    \caption{Hard instance for detecting a swap for $x=4$ and some $q=(s,q_B,t, q_C)$.}
    \label{fig:sp-direct}
  \end{figure}

  It is clear that $\pi_G(s,t) = 0$ and that we get $\pi_{G_q}(s,t) \geq \frac{(1-\alpha)\alpha}{x} \geq 2c\delta$ for all $q \in Q$, if we require $x\delta \leq 1$ when choosing $x$.
  Note that the assumption on degrees in~\Cref{sp-universal} is satisfied.
  Now,~\Cref{sp-universal} gives us a lower bound of~$\Omega(\abs{Q}/K)$.
  Since we currently assume that the model does not contain $\Adj$, we have $E_W = E$, (where $E$ denotes the edge set of $G$, not $G_q$) so $K = \max_{e \in E}\abs{\{q \in Q \mid e \cap E_q^\pm \neq \emptyset\}}$.
  This maximum is non-zero only for edges $e$ between $s$ and $B_{i_s}$ and between $C_{i_t}$ and $t$, so $K = \max\{\abs{C_{i_t}}, \abs{B_{i_s}}\} = x$.
  Since $Q = x^2$, we get a lower bound of $\Omega(x)$.
  We finish the proof for this model by choosing $x = \min\{d, 1/\delta\}$, and noting that $x \leq d$ and $x\delta \leq 1$ as required.

  Let us now consider the adjacency-list model with $\Jump$ and $\Adj$.
  We add two independent vertices to $G$ and for each $q = (s, q_B, t, q_C) \in Q$, we subdivide the edges $\{s, t\}$ and $\{q_B, q_C\}$ in $G_q$ using these two new vertices.
  We still have $\pi_G(s,t) = 0$ and $\pi_{G_q}(s,t) \geq \frac{(1-\alpha)^2\alpha}{2x} \geq 2c\delta$ for all $q \in Q$, if we require $x\delta \leq 1$ when choosing $x$.
  The proof of~\Cref{swap-lb} and thus~\Cref{sp-universal} carries through analogously with these subdivisions.
  The algorithm can find one of the two new vertices using $\Jump$, but this requires $\Omega(n)$ queries in expectation, which is more than the lower bound we wish to prove.
  So if we assume that the algorithm runs in time $o(n)$ on $G$ in expectation, we know that it with constant probability does not visists the two new vertices.
  When determining $K$ in~\Cref{swap-lb}, we can thus remove the two new vertices from $V_W$.
  When doing so, the halves of the subdivided edges do not appear in $E_W$, so we again get $K = x$, even though the model includes $\Adj$.
  We finish as above by setting $x = \min\{d, 1/\delta\}$.
\end{proof}

\subsection{Single source lower bounds}\label{sec:ss-lb}

We first combine \Cref{sp-separation} and \Cref{swap-lb} to obtain a general lemma for proving lower bounds for the single source problem.

\begin{lemma}\label{ss-universal}
  Let $\mc A$ be an algorithm solving the single source problem with approximation threshold $\delta$ in the adjacency-list model augmented with $U \subseteq \{\Jump, \NeighSorted, \Adj\}$.
  Let $G = (V, E)$ be a graph and $Q \subseteq V^4$ a set of swappable quadruples.
  If $\NeighSorted \in U$, assume that $d(q_1)=d(q_4)$ and $d(q_2)=d(q_3)$ for all $(q_1,q_2,q_3,q_4) \in Q$.
  Let $s \in V$ be a vertex such that for all $q \in Q$, there exists a vertex $t \in V$ with $\pi_G(s,t)=0$ and $\pi_{G_q}(s,t) > 2c\delta$.

  If $\Jump \in U$, let $V_W = V$, and otherwise, let $V_W \subseteq V$ be the union of components in $G$ containing $s$.
  If $\Adj \in U$, let $E_W = \binom{V_W}{2}$, and otherwise, let $E_W = \binom{V_W}{2} \cap E$.

  Let $K = \max_{e \in E_W}\abs{\{q \in Q \mid e \in E^\pm_q\}}$.
  Then, the expected query complexity of $\mc A$ on $G$ with source $s$ is $\Omega(\abs{Q}/K)$.
\end{lemma}
\begin{proof}
  The lower bound follows from~\Cref{swap-lb} with $W = (s)$ if we show $\P{\mc A(G, W) = \mc A(G_q, W)} \allowbreak\leq 2p_f$ for all $q \in Q$.
  For each $q \in Q$, let $t_q \in V$ be a vertex such that $\pi_G(s,t_q)=0$ and $\pi_{G_q}(s,t_q)>2c\delta$, which exists by assumption.
  For each $H \in \mc G(G, Q)$ and $q \in Q$, let $\hat\pi_H(s,t_q)$ be the estimate of $\pi_H(s,t_q)$ produced by $\mc A$ on $(H, W)$.
  Then for each $q \in Q$, the event $\mc A(G, W) = \mc A(G_q, W)$ implies that $\hat\pi_G(s,t_q)=\hat\pi_{G_q}(s,t_q)$, so either $\hat\pi_G(s,t_q) > c\delta$ or $\hat\pi_{G_q}(s,t_q) \leq c\delta$.
  Both of the latter two events happen with probability at most $p_f$ by~\Cref{sp-separation}, since $\mc A$ solves the single source problem.
  By a union bound, $\P{\mc A(G, W) = \mc A(G_q, W)} \leq 2p_f$.
\end{proof}

We now prove the lower bound for the single source problem in the adjacency-list model with all query types available by applying \Cref{ss-universal} to the hard instance used in the proof of \Cref{sp-ac-j-s-a}.
This gives the following average-case lower bound, which is tight even for worst-case analysis.

\begin{theorem}\label{ss-ac-j-s-a}
  Consider the adjacency-list model with $\Jump$, $\NeighSorted$, and $\Adj$.
  For any $n$ and $m$ with $n \leq m \leq n^2$ and any $\delta \in (0, 1]$, there exists a graph $G=(V,E)$ with $\Theta(n)$ vertices and $\Theta(m)$ edges, such that for any algorithm solving the single source problem, the expected running time on $G$ with approximation threshold $\delta$, averaging over all souces $s \in V$, is $\Omega(\min\{m, 1/\delta\})$.
\end{theorem}
\begin{proof}
  We assume $c \leq \frac14(1-\alpha)^3\alpha$.
  Take the construction of $G$ from the proof of~\Cref{sp-ac-j-s-a}.
  It suffices to show the lower bound for $\Theta(n)$ fixed sources, so fix $s \in A_{i_s}$ for some $i_s$.
  Choose $Q = A_{i_s} \times B_{i_s} \times C_1 \times D_1$, noting that this is a set of swappable quadruples.
  For each $q = (q_A, q_B, q_C, q_D) \in Q$, write $t_q = q_C$, and note that a random walk in $G_q$ starting at $s$ moves from $s$ to $B\setminus\{q_B\}$ with probability $(1-\alpha)(1-1/y)$, then to $q_A$ with probability $(1-\alpha)/x$, then to $t_q=q_C$ with probability $(1-\alpha)/y$, and then stops with probability $\alpha$.
  Then $\pi_G(s, t_q) = 0$ and $\pi_{G_q}(s, t_q) \geq (1-\alpha)^3\alpha(1-1/y)/(xy) \geq 2c\delta$ for all $q \in Q$ if we require $y \geq 2$ and $xy\delta \leq 1$ when choosing $x$ and $y$.
  By~\Cref{ss-universal}, we get a lower bound of $\Omega(\abs{Q}/K)$ where $K = \max_{e \in \binom V2}\abs{\{q \in Q \mid e \in E_q^\pm\}}$.
  For $e \in A_{i_s} \times B_{i_s}$, we have $\abs{\{q \in Q \mid e \in E_q^\pm\}} = \abs{C_1 \times D_1} = xy$ and similarly for $e \in (C_1 \times D_1) \cup (A_{i_s} \times C_1) \cup (B_{i_s} \times D_1)$.
  All other values of $e$ give $\abs{\{q \in Q \mid e \in E_q^\pm\}}=0$, so $K = xy$.
  Since $\abs{Q} = x^2y^2$, we get a lower bound of $\Omega(xy)$.

  We are now ready to choose $x$ and $y$, remembering to ensure $y \leq x \leq 2n$, $y \leq 2d$, and $xy\delta \leq 1$.

  \emph{Case 1:} For $0 < \delta \leq \frac{1}{m}$, set $x = 2n$ and $y = \floor{2d}$, giving a lower bound of $\Omega(m)$.

  \emph{Case 2:} For $\frac{1}{m} \leq \delta \leq \frac{1}{n}$, set $x = 2n$ and $y = \floor{2/(n\delta)}$, giving a lower bound of $\Omega(1/\delta)$.

  \emph{Case 3:} For $\frac{1}{n} \leq \delta \leq 1$, set $x = \floor{2/\delta}$ and $y = 2$, giving a lower bound of $\Omega(1/\delta)$.
\end{proof}

\subsection{Single target lower bounds}\label{sec:st-lb}

We first combine \Cref{sp-separation} and \Cref{swap-lb} to obtain a general lemma for proving lower bounds for the single target problem.

\begin{lemma}\label{st-universal}
  Let $\mc A$ be an algorithm solving the single target problem with approximation threshold $\delta$ in the adjacency-list model augmented with $U \subseteq \{\Jump, \NeighSorted, \Adj\}$.
  Let $G = (V, E)$ be a graph and $Q \subseteq V^4$ a set of swappable quadruples.
  If $\NeighSorted \in U$, assume that $d(q_1)=d(q_4)$ and $d(q_2)=d(q_3)$ for all $(q_1,q_2,q_3,q_4) \in Q$.
  Let $t \in V$ be a vertex such that for all $q \in Q$, there exists a vertex $s \in V$ with $\pi_G(s,t)=0$ and $\pi_{G_q}(s,t) > 2c\delta$.

  If $\Jump \in U$, let $V_W = V$, and otherwise, let $V_W \subseteq V$ be the union of components in $G$ containing $t$.
  If $\Adj \in U$, let $E_W = \binom{V_W}{2}$, and otherwise, let $E_W = \binom{V_W}{2} \cap E$.

  Let $K = \max_{e \in E_W}\abs{\{q \in Q \mid e \in E^\pm_q\}}$.
  Then, the expected query complexity of $\mc A$ on $G$ with target $t$ is $\Omega(\abs{Q}/K)$.
\end{lemma}
\begin{proof}
  The lower bound follows from~\Cref{swap-lb} with $W = (t)$ if we show $\P{\mc A(G, W) = \mc A(G_q, W)} \leq 2p_f$ for all $q \in Q$.
  For each $q \in Q$, let $s_q \in V$ be a vertex such that $\pi_G(s_q,t)=0$ and $\pi_{G_q}(s_q,t)>2c\delta$, which exists by assumption.
  For each $H \in \mc G(G, Q)$ and $q \in Q$, let $\hat\pi_H(s_q,t)$ be the estimate of $\pi_H(s_q,t)$ produced by $\mc A$ on $(H, W)$.
  Then for each $q \in Q$, the event $\mc A(G, W) = \mc A(G_q, W)$ implies that $\hat\pi_G(s_q,t)=\hat\pi_{G_q}(s_q,t)$, so either $\hat\pi_G(s_q,t) > c\delta$ or $\hat\pi_{G_q}(s_q,t) \leq c\delta$.
  Both of the latter two events happen with probability at most $p_f$ by~\Cref{sp-separation}, since $\mc A$ solves the single target problem.
  By a union bound, $\P{\mc A(G, W) = \mc A(G_q, W)} \leq 2p_f$.
\end{proof}

We now prove worst-case lower bounds for the single target problem by applying \Cref{st-universal} to explicit hard instances.
We start with the $\Adj$-only model.

\begin{theorem}\label{st-wc-a}
  Consider the adjacency-list model with $\Adj$.
  For any $n$ and $m$ with $1 \leq n \leq m \leq n^2$ and any $\delta \in (0,1]$, there exists a graph $G = (V,E)$ with $\Theta(n)$ vertices, $\Theta(m)$ edges, and a vertex $t \in V$, such that for any algorithm solving the single target problem, the expected running time on $G$ with target $t$ and approximation threshold $\delta$ is $\Omega(\min\{m, n/\delta\})$.
\end{theorem}
\begin{proof}
  For reference in future proofs, let us construction a graph with more parameters than neccessary for the current proof.
  Let $k$, $l$, $n_A$, $n_B$, $n_C$, and $n_D$ be positive integer parameters to be chosen later.
  Define disjoint sets $A_i$, $B_i$, $C_j$, and $D_j$ of size $n_A$, $n_B$, $n_C$, and $n_D$, respectively, for $i \in \{1,\ldots,k\}$ and $j \in \{1,\ldots,l\}$.
  Construct $G=(V,E)$ as the union of the complete bipartite graphs $K_{A_i,B_i}$ over $i \in \{1,\ldots,k\}$ and $K_{C_j,D_j}$ over $j \in \{1,\ldots,l\}$.
  Add to the construction an independt subgraph with $n$ vertices and $m$ edges.
  Requiring $\max\{kn_A,kn_B,ln_C,ln_D\} \leq n$ and $\max\{kn_An_B,ln_Cn_D\} \leq m$, we get a graph with $\Theta(n)$ vertices and $\Theta(m)$ edges.

  We assume $c \leq \frac14(1-\alpha)^3\alpha$.
  For the current proof, choose $k=l=n_A=n_B=1$, $n_C=n$, and $n_D=x$ where $x$ is a parameter to be chosen later, collapsing the above requirements to $x \leq d$.
  Fix $s \in B_1$ and $t \in D_1$.
  Let $Q = A_1 \times B_1 \times C_1 \times D_1$, noting that this is a set of swappable quadruples.
  We have $\pi_G(s,t)=0$ and $\pi_{G_q}(s,t) \geq (1-\alpha)^3\alpha(1-1/n)/x \geq 2c\delta$ if we require $x\delta \leq 1$ when choosing $x$, since we can assume $n \geq 2$ without loss of generality.
  To understand the middle term, imagine a walk moving from $s$ to $q_D \in D_1$ to anywhere in $C_1\setminus\{q_C\}$ to $t$, for $q=(q_A,q_B,q_C,q_D) \in Q$.
  By~\Cref{st-universal}, we get a lower bound of $\Omega(\abs{Q}/K)$, with $K = \max_{E_W}\abs{\{q \in Q \mid e \in E_q^\pm\}}$, where $E_W$ is the set of edges between $C_1$ and $D_1$.
  For each of these edges $e$, we have exactly one $q \in Q$ with $e \in E_q^\pm$, since $\abs{A}=\abs{B}=1$, so $K = 1$.
  Since $\abs{Q} = nx$, we get a lower bound of $\Omega(nx)$.
  We are now ready to choose $x$, remembering to ensure $x \leq d$ and $x\delta \leq 1$.

  \emph{Case 1:} For $0 < \delta \leq \frac{1}{d}$, set $x = \floor{d}$, giving a lower bound of $\Omega(m)$.

  \emph{Case 2:} For $\frac{1}{d} \leq \delta \leq 1$, set $x = \floor{1/\delta}$, giving a lower bound of $\Omega(n/\delta)$.
\end{proof}

We next consider the full model with $\Jump$, $\NeighSorted$, and $\Adj$.

\begin{theorem}\label{st-wc-j-s-a}
  Consider the adjacency-list model with $\Jump$, $\NeighSorted$, and $\Adj$.
  For any $n$ and $m$ with $1 \leq n \leq m \leq n^2$ and any $\delta \in (0,1]$, there exists a graph $G = (V,E)$ with $\Theta(n)$ vertices, $\Theta(m)$ edges, and a vertex $t \in V$, such that for any algorithm solving the single target problem, the expected running time on $G$ with target $t$ and approximation threshold $\delta$ is $\Omega(\min\{m, n/\delta^{1/2}\})$.
\end{theorem}
\begin{proof}
  We assume $c \leq \frac14(1-\alpha)^3\alpha$.
  Take the construction of $G$ from the proof of~\Cref{st-wc-a}, choosing $k=l=1$, $n_A=n_D=x$ and $n_B=n_C=n$, where $x$ is a parameter to be chosen later, which should satisfy $x \leq d$.
  (This is the same construction as in the proof of~\Cref{sp-wc-j-s-a} with $y=n$.)
  Fix $t \in D_1$.
  Choose $Q = A_1 \times B_1 \times C_1 \times D_1$, noting that this is a set of swappable quadruples.
  For each $q = (q_A, q_B, q_C, q_D) \in Q$, write $s_q = q_B$.
  For each $q \in Q$, we then have $\pi_G(s_q,t) = 0$ and $\pi_{G_q}(s_q,t) \geq (1-\alpha)^3(1-1/n)\alpha/x^2 \geq 2c\delta$ for all $q \in Q$, if we require $x^2\delta \leq 1$ when choosing $x$, since we can assume $n \geq 2$ without loss of generality.
  By~\Cref{st-universal}, we get a lower bound of $\Omega(\abs{Q}/K)$, with $K=\max_{e \in \binom{V}{2}}\abs{\{q \in Q \mid e \in E_q^\pm\}}$.
  It is easy to see that $K = xn$, so since $\abs{Q} = x^2n^2$, we get a lower bound of $\Omega(xn)$.
  We now choose $x$, remembering to ensure $x \leq d$ and $x^2\delta \leq 1$.

  \emph{Case 1:} For $0 < \delta \leq \frac{1}{d^2}$, set $x = \floor{d}$, giving a lower bound of $\Omega(m)$.

  \emph{Case 2:} For $\frac{1}{d^2} \leq \delta \leq 1$, set $x = \floor{1/\delta^{1/2}}$, giving a lower bound of $\Omega(n/\delta^{1/2})$.
\end{proof}

We now turn to average-case lower bounds for the single target problem.
We first prove a lower bound of $\Omega(d)$ for the model with $\Jump$ and $\Adj$.

\begin{lemma}\label{st-ac-degree}
  Consider the adjacency-list model with $\Jump$ and $\Adj$.
  For any $n$ and $m$ with $1 \leq n \leq m \leq n^2$ and any $\delta \in (0,1]$, there exists a graph $G = (V,E)$ with $\Theta(n)$ vertices, $\Theta(m)$ edges, such that for any algorithm solving the single target problem, the expected running time on $G$ with approximation threshold $\delta$, averaging over all targets $t \in V$, is $\Omega(d)$.
\end{lemma}
\begin{proof}
  We assume $c \leq \frac12(1-\alpha)\alpha$.
  Take the construction of $G$ from the proof of~\Cref{st-wc-a}, choosing $k=l=n_A=n_B=1$, $n_C=d$, and $n_D=n$, noting that this satisfies the requirements.
  It suffices to show the lower bound of $\Theta(n)$ targets, so fix $t \in D_1$.
  Choose $Q = A_1 \times B_1 \times C_1 \times \{t\}$, noting that this is a set of swappable quadrangles.
  Write $s \in B_1$, and note that $\pi_G(s,t)=0$ and $\pi_{G_q}(s,t) \geq (1-\alpha)\alpha \geq 2c\delta$.
  Analogously to the proof of~\Cref{sp-ac-j-s-xor-a}, specifically the case handling the model with $\Adj$, we can assume that the algorithm does not visit $A_1$ and $B_1$, since they have constant size and we are trying to prove a lower bound below $O(n)$.
  By~\Cref{st-universal}, we get a lower bound of $\Omega(\abs{Q}/K)$, where following this analogy, we again get $K=1$, as we can assume $A_1$ and $B_1$ are not in $V_W$.
  Since $\abs{Q} = d$, we get a lower bound of $\Omega(d)$.
\end{proof}

We now prove an average-case lower bound in the $\Adj$-only model.

\begin{theorem}\label{st-ac-a}
  Consider the adjacency-list model with $\Adj$.
  For any $n$ and $m$ with $1 \leq n \leq m \leq n^2$ and any $\delta \in (0,1]$, there exists a graph $G = (V,E)$ with $\Theta(n)$ vertices, $\Theta(m)$ edges, such that for any algorithm solving the single target problem, the expected running time on $G$ with approximation threshold $\delta$, averaging over all targets $t \in V$, is $\Omega(\min\{m, d/\delta, 1/\delta^{2}+d\})$.
\end{theorem}
\begin{proof}
  We assume $c \leq \frac14(1-\alpha)^3\alpha$.
  Take the construction of $G$ from the proof of~\Cref{st-wc-a}, choosing $k=n_A=n_B=1$, $l=\floor{n/x}$, $n_C=y$, and $n_D=x$, where $x$ and $y$ are parameters to be chosen later, which should satisfy $y \leq x \leq 2n$ and $y \leq 2d$.
  It suffices to show the lower bound for $\Theta(n)$ fixed targets, so fix $t \in D_{i_t}$ for some $i_t$.
  Choose $Q = A_1 \times B_1 \times C_{i_t} \times D_{i_t}$, noting that this is a set of swappable quadruples.
  Write $s \in B_1$, and note that $\pi_G(s,t) = 0$ and $\pi_{G_q}(s,t) \geq (1-\alpha)^3\alpha(1-1/y)/x \geq 2c\delta$ for all $q \in Q$, if we require $y \geq 2$ and $x\delta \leq 1$ when choosing $x$.
  By~\Cref{st-universal}, we get a lower bound of $\Omega(\abs{Q}/K)$, where $K=\max_{e \in E_W}\abs{\{q \in Q \mid e \in E_q^\pm\}}$.
  Note that $V_W=C_{i_t}\cup D_{i_t}$, so $E_W$ is the set of edges $e$ between $C_{i_t}$ and $D_{i_t}$, all of which have $\abs{\{q \in Q \mid e \in E_q^\pm\}}=\abs{A_1\times B_1}=1$, so $K=1$.
  Since $Q=xy$, we get a lower bound of $\Omega(xy)$.
  We now choose $x$ and $y$, remembering to ensure $2 \leq y \leq x \leq 2n$, $y \leq 2d$, and $x\delta \leq 1$.

  \emph{Case 1:} For $0 < \delta \leq \frac{1}{n}$, set $x = 2n$ and $y = \floor{2d}$, giving a lower bound of $\Omega(m)$.

  \emph{Case 2:} For $\frac{1}{n} \leq \delta \leq \frac{1}{d}$, set $x = \floor{2/\delta}$ and $y = \floor{2d}$, giving a lower bound of $\Omega(d/\delta)$.

  \emph{Case 3:} For $\frac{1}{d} \leq \delta \leq 1$, set $x = y = \floor{2/\delta}$, giving a lower bound of $\Omega(1/\delta^2)$.

  Finally,~\Cref{st-ac-degree} gives a lower bound of $\Omega(d)$.
\end{proof}

We next consider the model with $\Jump$ and $\Adj$.

\begin{theorem}\label{st-ac-j-a}
  Consider the adjacency-list model with $\Jump$ and $\Adj$.
  For any $n$ and $m$ with $1 \leq n \leq m \leq n^2$ and any $\delta \in (0,1]$, there exists a graph $G = (V,E)$ with $\Theta(n)$ vertices, $\Theta(m)$ edges, such that for any algorithm solving the single target problem, the expected running time on $G$ with approximation threshold $\delta$, averaging over all targets $t \in V$, is $\Omega(\min\{m, (m/\delta)^{1/2}, d/\delta, (n/\delta)^{2/3}+d, 1/\delta^{2}+d\})$.
\end{theorem}
\begin{proof}
  We assume $c \leq \frac14(1-\alpha)^3\alpha$.
  Take the construction of $G$ from the proof of~\Cref{st-wc-a}, choosing $k=\floor{n/x}$, $l=\floor{n/z}$, $n_A=n_B=x$, $n_C=y$, and $n_D=z$, where $x$, $y$, and $z$ are parameters to be chosen later, which should satisfy $y \leq z \leq 2n$ and $\max\{x, y\} \leq 2d$.
  It suffices to show the lower bound for $\Theta(n)$ fixed targets, so fix $t \in D_{i_t}$ for some $i_t$.
  Choose $Q = \bigcup_{i=1}^{k}A_i \times B_i \times C_{i_t} \times D_{i_t}$, noting that this is a set of swappable quadruples.
  For each $q = (q_A, q_B, q_C, q_D) \in Q$, write $s_q = q_B$, and note that $\pi_G(s_q,t)=0$ and $\pi_{G_q}(s_q,t) \geq (1-\alpha)^3\alpha(1-1/y)/(xz) \geq 2c\delta$ if we require $y \geq 2$ and $xz\delta \leq 1$ when choosing $x$, $y$, and $z$.
  By~\Cref{st-universal}, we get a lower bound of $\Omega(\abs{Q}/K)$.
  By looking at the four edges of $E_q^\pm$ seperately, we see that $K = \Theta(\max\{zy,nx,xz,xy\})=\Theta(\max\{zy,nx\})$.
  Since $\abs{Q} = \Theta(nxyz)$, we get a lower bound of $\Omega(\min\{nx, yz\})$.
  We now choose $x$, $y$, and $z$, remembering to ensure $2 \leq y \leq z \leq 2n$, $\max\{x, y\} \leq 2d$, and $xz\delta \leq 1$.

  \emph{Case 1:} For $0 < \delta \leq \frac{1}{m}$, set $x = y = \floor{2d}$ and $z = 2n$, giving a lower bound of $\Omega(m)$.

  \emph{Case 2:} For $\frac{1}{m} \leq \delta \leq \min\{\frac{d}{n}, \frac{n}{d^3}\}$, set $x = \floor{(d/(n\delta))^{1/2}}$, $y = \floor{2d}$, and $z = \floor{2(n/(d\delta))^{1/2}}$, giving a lower bound of $\Omega((m/\delta)^{1/2})$.

  \emph{Case 3:} For $\frac{d}{n} \leq \delta \leq \frac{1}{d}$, set $x = 1$, $y = \floor{2d}$, and $z = \floor{2/\delta}$, giving a lower bound of $\Omega(d/\delta)$.

  \emph{Case 4:} For $\frac{n}{d^3} \leq \delta \leq \frac{1}{n^{1/2}}$, set $x=\floor{n^{-1/3}\delta^{-2/3}}$ and $y=z=\floor{2n^{1/3}\delta^{-1/3}}$, giving a lower bound of $\Omega((n/\delta)^{2/3})$.

  \emph{Case 5:} For $\max\{\frac{1}{d}, \frac{1}{n^{1/2}}\} \leq \delta \leq 1$, set $x = 1$, $y=z=\floor{2/\delta}$, giving a lower bound of $\Omega(1/\delta^2)$.

  Finally,~\Cref{st-ac-degree} gives a lower bound of $\Omega(d)$.
\end{proof}

We finally prove the average-case lower bound for the model with $\Jump$, $\NeighSorted$, and $\Adj$.

\begin{theorem}\label{st-ac-j-s-a}
  Consider the adjacency-list model with $\Jump$, $\NeighSorted$, and $\Adj$
  For any $n$ and $m$ with $1 \leq n \leq m \leq n^2$ and any $\delta \in (0,1]$, there exists a graph $G = (V,E)$ with $\Theta(n)$ vertices, $\Theta(m)$ edges, such that for any algorithm solving the single target problem, the expected running time on $G$ with approximation threshold $\delta$, averaging over all targets $t \in V$, is $\Omega(\min\{m, 1/\delta\})$.
\end{theorem}
\begin{proof}
  We assume $c \leq \frac{1}{8}(1-\alpha)^4\alpha$.
  Take the construction of $G$ from the proof of~\Cref{st-wc-a}, choosing $k=1$, $l=\floor{n/x}$, $n_A=n_D=x$, and $n_B=n_C=y$, where $x$ and $y$ are parameters to be chosen later, which should satisfy $y \leq x \leq 2n$ and $y \leq 2d$.
  It suffices to show the lower bound for $\Theta(n)$ fixed targets, so fix $t \in D_{i_t}$ for some $i_t$.
  Choose $Q = A_1 \times B_1 \times C_{i_t} \times D_{i_t}$, noting that this is a set of swappable quadruples.
  For each $q = (q_A, q_B, q_C, q_D) \in Q$, write $s_q = q_A$, and note that $\pi_G(s_q,t)=0$ and $\pi_{G_q}(s_q,t) \geq (1-\alpha)^4\alpha(1-1/x)(1-1/y)/(yx) \geq 2c\delta$ if we require $2 \leq \min\{x,y\}$ and $xy\delta \leq 1$ when choosing $x$ and $y$.
  By~\Cref{st-universal}, we get a lower bound of $\Omega(\abs{Q}/K)$, and $K = xy$, analogously to~\Cref{ss-ac-j-s-a} or easily verified.
  Since $\abs{Q} = x^2y^2$, we get a lower bound of $\Omega(xy)$.
  We now choose $x$ and $y$, remembering to ensure $2 \leq y \leq x \leq 2n$, $y \leq 2d$, and $xy\delta \leq 1$.

  \emph{Case 1:} For $0 < \delta \leq \frac{1}{m}$, set $x=2n$ and $y=\floor{2d}$, giving a lower bound of $\Omega(m)$.

  \emph{Case 2:} For $\frac{1}{m} \leq \delta \leq \frac{1}{n}$, set $x=2n$ and $y=\floor{2/(n\delta)}$, giving a lower bound of $\Omega(1/\delta)$.

  \emph{Case 3:} For $\frac{1}{n} \leq \delta \leq 1$, set $x=\floor{2/\delta}$ and $y=2$, giving a lower bound of $\Omega(1/\delta)$.
\end{proof}

\subsection{Single node lower bounds}\label{sec:sn-lb}

We need a seperation similar to~\Cref{sp-separation} for the single node problem. 

\begin{lemma}\label{sn-separation}
  Let $G$ be a graph containing a vertex $t$.
  Let $\hat\pi(t) \in [0, 1]$ be an estimate of $\pi(t)$ satisfying $\abs{\hat\pi(t)-\pi(t)} \leq c\pi(s,t)$.
  Then, for any $x \in [0, 1]$,
  \begin{itemize}
    \item $\hat\pi(t) \leq (1+c)x$ if $\pi(t) \leq x$, and
    \item $\hat\pi(t) > (1+c)x$ if $\pi(t) \geq (1+4c)x$.
  \end{itemize}
\end{lemma}
\begin{proof}
  We assume $c < \frac12$.
  If $\pi(t) \leq x$, we have $\hat\pi(t) \leq (1+c)\pi(t) \leq (1+c)x$.
  If $\pi(t) \geq (1+4c)x$, we have $\hat\pi(t) \geq (1-c)\pi(t) \geq (1-c)(1+4c)x > (1+c)x$, where the last inequality follows easily from $c < \frac12$.
\end{proof}

We next combine \Cref{sn-separation} and \Cref{swap-lb} to obtain a general lemma for proving lower bounds for the single node problem.

\begin{lemma}\label{sn-universal}
  Let $\mc A$ be an algorithm solving the single node problem in the adjacency-list model augmented with $U \subseteq \{\Jump, \NeighSorted, \Adj\}$.
  Let $G = (V, E)$ be a graph and $Q \subseteq V^4$ a set of swappable quadruples.
  If $\NeighSorted \in U$, assume that $d(q_1)=d(q_4)$ and $d(q_2)=d(q_3)$ for all $(q_1,q_2,q_3,q_4) \in Q$.
  Let $t \in V$ be a vertex with $\pi_{G_q}(t) \geq (1+4c)\pi_G(t)$ for all $q \in Q$.

  If $\Jump \in U$, let $V_W = V$, and otherwise, let $V_W \subseteq V$ be the union of components in $G$ containing $s$ or $t$.
  If $\Adj \in U$, let $E_W = \binom{V_W}{2}$, and otherwise, let $E_W = \binom{V_W}{2} \cap E$.

  Let $K = \max_{e \in E_W}\abs{\{q \in Q \mid e \in E^\pm_q\}}$.
  Then, the expected query complexity of $\mc A$ on $G$ with node $t$ is $\Omega(\abs{Q}/K)$.
\end{lemma}
\begin{proof}
  This follows from~\Cref{swap-lb} with $W = (t)$ if we show $\P{\mc A(G, W) = \mc A(G_q, W)} \leq 2p_f$ for all $q \in Q$.
  Let $\hat\pi_G(t)$ be the estimate produced by $\mc A$ on $(G, W)$ and $\hat\pi_{G_q}(t)$ be the estimate produced by $\mc A$ on $(G_q, W)$.
  Then for each $q \in Q$, the event $\mc A(G, W) = \mc A(G_q, W)$ implies that $\hat\pi_G(t)=\hat\pi_{G_q}(t)$, so either $\hat\pi_G(t) > (1+c)x$ or $\hat\pi_{G_q}(t) \leq (1+c)x$.
  Both of the latter two events happen with probability at most $p_f$ by~\Cref{sn-separation} with $x=\pi_G(t)$, since $\mc A$ solves the single node problem.
  By a union bound, $\P{\mc A(G, W) = \mc A(G_q, W)} \leq 2p_f$.
\end{proof}

We now prove worst-case lower bounds for the single node problem by applying \Cref{sn-universal} to an explicit hard instance.
In \cite{BackMC}, an $\Omega(m^{1/2})$ lower bound is shown for the standard adjacency-list model.
We further show that the same bound continues to hold for any model that does not include all three query types.

\begin{theorem}\label{sn-wc-not-all}
  Consider the adjacency-list model with the additional queries of a proper subset $U\subsetneq \{\Jump, \NeighSorted, \Adj\}$.
  For any $n$ and $m$ with $1 \leq n \leq m \leq n^2$, there exists a graph $G = (V,E)$ with $\Theta(n)$ vertices, $\Theta(m)$ edges, such that for any algorithm solving the single node problem, the expected running time on $G$ is $\Omega(m^{1/2})$.
\end{theorem}
\begin{proof}
  We assume $c < \frac{1}{84}(1-\alpha)^2\alpha$.
  Let $x$ be a positive integer parameter to be chosen later.
  Define disjoint sets $A$, $B$, $C$, $D$, and $F$ with $\abs{A}=\abs{B}=\abs{C}=\abs{D}=x$ and $\abs{F}=2x$.
  Fix a vertex $b \in B$.
  Construct $G=(V, E)$ as the union of complete bipartite graphs $K_{A,\{b\}}$, $K_{B, C}$, and $K_{D, F}$.
  Add to the construction an independent subgraph with $n$ vertices and $m$ edges.
  Requiring $x^2 \leq m$, we get a graph with $\Theta(n)$ vertices and $\Theta(m)$ edges.

  Fix $t \in F$ and let $Q = \{b\} \times C \times \{t\} \times D$, noting that this is a set of swappable quadrangles.
  See~\Cref{fig:sn} for an illustration of $G$ and $G_q$ for some $q \in Q$.

  \begin{figure}[h]
    \centering
    \begin{tikzpicture}[
  styleN1/.style={circle, draw, inner sep=0, minimum size=18pt, line width=0.5pt},
  styleD1/.style={line width=0.5pt},
  styleD2/.style={line width=0.7pt, densely dashed, color=red},
  styleD3/.style={line width=0.7pt, color=blue},
  >=latex
  ]
  \pgfmathsetmacro{\K}{3};
  \pgfmathsetmacro{\L}{\K*2};
  \pgfmathsetmacro{\labelOffset}{.3}
  \pgfmathsetmacro{\layerZero}{0}
  \pgfmathsetmacro{\layerOne}{\layerZero-1}
  \pgfmathsetmacro{\layerTwo}{\layerOne-1}
  \pgfmathsetmacro{\layerThree}{\layerTwo-1}
  \pgfmathsetmacro{\layerFour}{\layerThree-1.5}
  \pgfmathsetmacro{\swapB}{2};
  \pgfmathsetmacro{\swapC}{3};
  \pgfmathsetmacro{\swapD}{2};
  \pgfmathsetmacro{\swapF}{3};
  \foreach \i in {1,...,\K} {
    \pgfmathsetmacro{\xi}{\i-1-(\K-1) / 2}
    \node[styleN1] (vA\i) at (\xi, \layerZero) {};
    \node[styleN1] (vB\i) at (\xi, \layerOne) {\ifthenelse{\i = \swapB}{$b$}{}};
    \node[styleN1] (vC\i) at (\xi, \layerTwo) {\ifthenelse{\i = \swapC}{$q_C$}{}};
    \node[styleN1] (vD\i) at (\xi, \layerThree) {\ifthenelse{\i = \swapD}{$q_D$}{}};
  }
  \foreach \i in {1,...,\L} {
    \pgfmathsetmacro{\xi}{\i-1-(\L-1) / 2}
    \node[styleN1] (vF\i) at (\xi, \layerFour) {\ifthenelse{\i = \swapF}{$t$}{}};
  }
  \foreach \i in {1,...,\K} {
    \draw[styleD1] (vA\i) -- (vB\swapB);
  }
  \foreach \i in {1,...,\K} {
    \foreach \j in {1,...,\K} {
      \ifthenelse{\i = \swapB \AND \j = \swapC}{\draw[styleD2] (vB\i) -- (vC\j)}{\draw[styleD1] (vB\i) -- (vC\j)};
    }
  }
  \foreach \i in {1,...,\K} {
    \foreach \j in {1,...,\L} {
      \ifthenelse{\i = \swapD \AND \j = \swapF}{\draw[styleD2] (vD\i) -- (vF\j)}{\draw[styleD1] (vD\i) -- (vF\j)};
    }
  }
  \draw[styleD3] (vB\swapB) to[out=-120,in=90] (vF\swapF);
  \draw[styleD3] (vC\swapC) to (vD\swapD);
  \node (VA) at (\K/2+\labelOffset, \layerZero) {$A$};
  \node (VB) at (\K/2+\labelOffset, \layerOne) {$B$};
  \node (VC) at (\K/2+\labelOffset, \layerTwo) {$C$};
  \node (VD) at (\K/2+\labelOffset, \layerThree) {$D$};
  \node (VF) at (\L/2+\labelOffset, \layerFour) {$F$};
\end{tikzpicture}
    \caption{Hard instance for detecting a swap for $x=3$, and some $q=(b,q_C,t,q_D)$.
    Swapping the {\color{red} red} edge pair with the {\color{blue} blue} edge pair, walks starting in $A$ cause a substantial increase in the PageRank of $t$.
    An algorithm has to distinguish between the two cases to satisfy the estimation requirement.
    }
    \label{fig:sn}
  \end{figure}
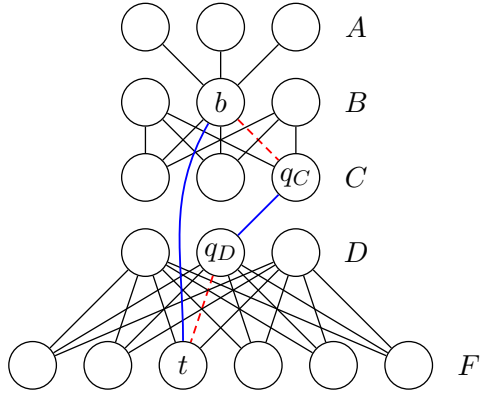
  In $G$, a walk can only end in $t$ by starting in $D$ or $F$, and a walk starting in a uniformly random vertex of $D \cup F$ will end at $t$ with probability at most $\frac{1}{2x}$ by symmetry, hence $\pi_G(t) \leq \frac{\abs{D \cup F}}{\abs{V}2x} = \frac{3x}{(6x+n)2x} \leq \frac{3}{2n}$.
  In $G_q$ for any $q \in Q$, a walk starting at a uniformly random vertex will start in $A$ and end in $t$ with probability at least $\frac{\abs{A}}{\abs{V}}\frac{(1-\alpha)^2\alpha}{d(b)} = \frac{x(1-\alpha)^2\alpha}{(6x+n)2x} \geq 4c\pi_G(t)$.
  To compute $\pi_{G_q}(t)$ exactly, we can solve the set of equations for the markov chain of a random walk in $G_q$ (e.g.\ using computer algebra software).
  We will not write this out here, as it would be quite opaque, but the above bound of $4c\pi_G(t)$ sketches why we can get $\pi_{G_q}(t) \geq (1+4c)\pi_G(t)$ with an appropriate bound on $c$.
  We can now apply~\Cref{sn-universal} in two cases.

  If $\Jump \not\in U$ or $\Adj \not\in U$, we get a lower bound of $\Omega(\abs{Q}/K)$ where the maximum in the definition of $K$ is taken over edges between $D$ and $F$, hence $K = \abs{C} = x$.
  Since $\abs{Q} = x^2$, we get a lower bound of $\Omega(x)$.
  Setting $x = \floor{m^{1/2}}$ finishes this case.

  If $\NeighSorted \not\in U$, we subdivide the edges $\{b, t\}$ and $\{q_C, q_D\}$ in $G_q$ as in~\Cref{sp-ac-j-s-xor-a}, giving a lower boun of $\Omega(\min\{n,\abs{Q}/K\})$ where we get $K=\max\{\abs{C},\abs{D}\} = x$, and finish as above.
\end{proof}

We next consider the full model with $\Jump$, $\NeighSorted$, and $\Adj$.

\begin{theorem}\label{sn-wc-all}
  Consider the adjacency-list model with $\Jump$, $\NeighSorted$, and $\Adj$.
  For any $n$ and $m$ with $1 \leq n \leq m \leq n^2$, there exists a graph $G = (V,E)$ with $\Theta(n)$ vertices, $\Theta(m)$ edges, such that for any algorithm solving the single node problem, the expected running time on $G$ is $\Omega(n^{1/2})$.
\end{theorem}
\begin{proof}
  Use the construction of~\Cref{sn-wc-not-all} with $x=\floor{n^{1/2}}$.
  In expectation, it takes $\Theta\p{\frac{\abs{A \cup B \cup C}}{\abs{V}}} = \Theta(n^{1/2})$ queries to find a vertex in $A$, $B$, or $C$.
  So we can assume that the algorithm with constant probability finds no such vertex, in which case~\Cref{sn-universal} can be applied with $V_W = D \cup F$, giving a lower bound of $\Omega(\abs{Q}/K)$ with $K$ using $E_W=\binom{V_W}{2}$.
  We have $\abs{Q} = x^2$ and $K = \abs{C} = x$, so we get a lower bound of $\Omega(x)=\Omega(n^{1/2})$.
\end{proof}

We now turn to average-case lower bounds.
We again use the same hard instance, but with many disjoint copies to make a constant fraction of targets hard.
For the model that does not include all three query types, we get the following lower bound.

\begin{theorem}\label{sn-ac-not-all}
  Consider the adjacency-list model with the additional queries of a proper subset $U\subsetneq \{\Jump, \NeighSorted, \Adj\}$.
  For any $n$ and $m$ with $1 \leq n \leq m \leq n^2$, there exists a graph $G = (V,E)$ with $\Theta(n)$ vertices, $\Theta(m)$ edges, such that for any algorithm solving the single node problem, the expected running time on $G$, averaging over all targets $t\in V$, is $\Omega(d)$.
\end{theorem}
\begin{proof}
  Take the construction of~\Cref{sn-wc-not-all} with $x=d$, but with $\floor{n/d}$ copies of $D$, $F$, and $K_{D,F}$.
  This graph still has $\Theta(n)$ vertices and $\Theta(m)$ edges.
  Fixing $t$ in any copy of $F$, the proof follows through, giving a lower bound of $\Omega(x)=\Omega(d)$.
  Since there are $\Theta(n)$ such choices of $t$, the lower bound holds on average over all vertices.
\end{proof}

Finally, for the model with all three query types, we show the following lower bound.

\begin{theorem}\label{st-ac-all}
  Consider the adjacency-list model with $\Jump$, $\NeighSorted$, and $\Adj$.
  For any $n$ and $m$ with $1 \leq n \leq m \leq n^2$, there exists a graph $G = (V,E)$ with $\Theta(n)$ vertices, $\Theta(m)$ edges, such that for any algorithm solving the single node problem the expected running time on $G$, averaging over all targets $t\in V$, is $\Omega(\min\{d, n^{1/2}\})$.
\end{theorem}
\begin{proof}
  Take the construction of~\Cref{sn-wc-not-all} with $x=\min\{d, n^{1/2}\}$, but with $\floor{n/d}$ copies of $D$, $F$, and $K_{D,F}$.
  This graph still has $\Theta(n)$ vertices and $\Theta(m)$ edges.
  Since there are $\Theta(n)$ vertices in total across all copies of $F$, it suffices to show the lower bound for a fixed $t$ in one of these copies, so fix such a $t$.
  Since $x \leq n^{1/2}$, the proof of~\Cref{sn-wc-all} follows through, giving a lower bound of $\Omega(x)=\Omega(\min\{d, n^{1/2}\})$.
\end{proof}

\bibliography{references}

\appendix
\section{Appendix}\label{sec:appendix}

\subsection{Pseudocode}\label{sec:app-pseudocode}

\begin{algorithm}[H]
    \caption{$\BackPushNew(t,r_{\max})$ }
    \label{alg:bacwards-push-improved}
    \textbf{Input:} graph $G$, target vertex $t$, residual threshold $r_{\max}$\\
    \textbf{Output:} reserves $p()$ and residuals $r()$
    \begin{algorithmic}[1]
    \State $r(), p()\gets$ \text{empty dictionary with default value $0$}
    \State $p(t) \gets \alpha$
    \For{$i$ from $1$ to $\Deg(t)$} \text{\footnotesize\textcolor{gray}{// Can be done in $O(|\{u\in \mc N(t) \mid d(u)\leq 1/{r_{max}}\}|)$ time, if $\NeighSorted$ is available}}
        \State $u \gets \Neigh(t,i)$
        \If{$\Deg(u) \leq 1/r_{\max}$}
            \State $\residue(u) \gets \residue(u) + (1-\alpha) / \Deg(u)$
        \EndIf
    \EndFor
    \While{ exists $v$ with $r(v) > r_{\max}$}
        \State $\residue \gets \residue(v)$
        \State $\residue(v)\gets 0$
        \State $\reserve(v)\gets \reserve(v) + \alpha \residue$
        \For{$i$ from $1$ to $\Deg(v)$}
            \State $u \gets \Neigh(v,i)$
            \State $\residue(u) \gets \residue(u) + (1-\alpha)\residue(v) / \Deg(u)$
        \EndFor
    \EndWhile
    \State \Return $\residue()$ and $\reserve()$
    \end{algorithmic}
\end{algorithm}

\begin{algorithm}[H]
    \caption{$\BiPPRAvg(t,r_{\max},n_r)$ }
    \label{alg:bacwards-push-improved-sp}
    \textbf{Input:} target vertex $t$, residual threshold $r_{\max}$, number of random walks $n_r$ \\
    \textbf{Output:} estimate $\pih(s,t)$ of $\pi(s,t)$
    \begin{algorithmic}[1]
    \State $r(), p()\gets \BackPushNew(t,r_{\max})$ 
    \For{$i$ from $1$ to $\Deg(t)$}
        \State $u \gets \Neigh(t,i)$
        \If{$\Deg(u) > 1/r_{\max}$}
            \State $\residue(u) \gets \residue(u) + (1-\alpha) / \Deg(u)$
        \EndIf
    \EndFor
    \State $\pih(s,t) \gets p(s)$
    \For{$i\in [n_r]$}
        \State $u \leftarrow  \RW(G, s, \alpha)$ \label{alg-line:np-rw-3}
        \State $\pih(s,t) \gets \pih(s,t) + \frac{r(u)}{n_r}$
    \EndFor
    \Return $\pih(s,t)$
    \end{algorithmic}
\end{algorithm}

\begin{algorithm}[H]
\caption{$\RW(G, s,\alpha)$ }
\label{alg:mc}
\textbf{Input:} Undirected graph $G=(V,E)$, source vertex $s\in V$, stopping probability $\alpha$\\
\textbf{Output:} Sampled vertex $v$, where $v$ is sampled w.p. $\pi(s,v)$
\begin{algorithmic}[1]
    \State $v\gets s$
    \While{True}
        \State with probability $\alpha$ \textbf{return} $v$
        \State $v \gets \Neigh(v, \text{randint}(\Deg(v)))$
    \EndWhile
\end{algorithmic}
\end{algorithm}

\subsection{Missing proofs}\label{sec:missing-proofs}

\begin{proof}[Proof of \Cref{lem:bpnew-invariant}]
    We start by showing that the invariant is true after initializing the residual and reserve values.
    By \Cref{eq:rec-target} we have
    \begin{align*}
        \pi(u,t) &= \alpha \Ind{u=t} + \sum_{v\in \mc N(t)} \frac{(1-\alpha)\pi(u,v)}{d(v)} 
        = p(u) + \sum_{v \in V}\p{r(v)+\frac{1-\alpha}{d(v)}\Ind{v\in X}} \pi(u,v).
    \end{align*}
    The latter equality follows from observing that if $u=t$ then $p(u)=\alpha$, and if $u\in X$ then $r(v)=0$ and otherwise $r(u)=(1-\alpha)/d(u)$.
    Assume that we do a push at a vertex $w$. 
    Let $r'()$ and $p'()$ be the values before the push. 
    Note that $p(u)=p'(u) + \Ind{u=w} \alpha r'(w)$, and $r(u)=r'(u)$ for all $u\notin \mc N[w]$.
    Using this, we get for any $u\in V$:
    \begin{align}
        \pi(u,t) &= p(u) + \sum_{v \in V}\p{r(v)+\frac{1-\alpha}{d(v)}\Ind{v\in X}} \pi(u,v)
        = p'(u) + \sum_{v \in V\setminus \mc N[w]}\p{r'(v)+\frac{1-\alpha}{d(v)}\Ind{v\in X}} \pi(u,v)\nonumber \\ 
        &+ \sum_{v \in N[w]}\p{r(v)+\frac{1-\alpha}{d(v)}\Ind{v\in X}} \pi(u,v) + \Ind{u=w} \alpha r'(w)\label{eq:bpnew-1}
    \end{align}
    We further have that $r(w)=0$ and if $u\in N(w)$ then $r(u)=r'(u) + \frac{1-\alpha}{d(u)} r'(w)$.
    Using this, we get
    \begin{align*}
        & \sum_{v \in N[w]}\p{r(v)+\frac{1-\alpha}{d(v)}\Ind{v\in X}} \pi(u,v) + \Ind{u=w} \alpha r'(w)\\
        &= \sum_{v \in N[w]}\p{r'(v)\Ind{v\neq w}+\frac{1-\alpha}{d(v)}\Ind{v\in X}} \pi(u,v) + r'(w)\p{\sum_{v \in N(w)}\frac{(1-\alpha)\pi(u,v)}{d(v)} + \Ind{u=w} \alpha}\\
        &= \sum_{v \in N[w]}\p{r'(v)\Ind{v\neq w}+\frac{1-\alpha}{d(v)}\Ind{v\in X}} \pi(u,v) + r'(w) \pi(u,w)\\
        &= \sum_{v \in N[w]}\p{r'(v)+\frac{1-\alpha}{d(v)}\Ind{v\in X}} \pi(u,v)
    \end{align*}
    Combining this with \Cref{eq:bpnew-1} finishes the proof.
\end{proof}

\end{document}